\documentclass[11pt]{scrartcl}

\usepackage[hang,flushmargin]{footmisc}
\usepackage[toc,page]{appendix}
\ifx\href\undefined\else\hypersetup{linktocpage=true}\fi

\usepackage{mathrsfs} 
\usepackage{graphicx}
\usepackage{amsmath}
\usepackage{amssymb}
\usepackage{amsthm}
\usepackage{times}
\usepackage{euler}

\theoremstyle{plain}
\newtheorem{theorem}{Theorem}
\newtheorem{definition}{Definition}
\newtheorem{corollary}{Corollary}

\newcommand{\End}{\mathrm{End}}
\newcommand{\Alt}{\mathrm{Alt}}
\newcommand{\gup}{g_{\scriptscriptstyle\uparrow}}
\newcommand{\gdown}{g_{\scriptscriptstyle\downarrow}}
\newcommand{\integers}{\mathbb{Z}}
\newcommand{\reals}{\mathbb{R}}
\newcommand{\EMT}{\mathbf{T}}
\newcommand{\Section}{\Gamma}
\newcommand{\id}{\mathrm{id}}
\newcommand{\EMform}{\mathcal{T}}
\newcommand{\CurrentForm}{\mathcal{J}}

\newcommand{\MM}{\mathfrak{M}\,}
\newcommand{\Wedge}{{\textstyle\bigwedge}}
\newcommand{\Lie}{\mathrm{Lie}}
\newcommand{\GL}{\mathrm{GL}}
\newcommand{\Ad}{\mathrm{Ad}}
\newcommand{\Poin}{\mathrm{Poin}}
\newcommand{\Lor}{\mathrm{Lor}}
\newcommand{\Trans}{\mathrm{Trans}}
\newcommand{\Aff}{\mathrm{Aff}}

\begin{document}

\title{Laue's Theorem Revisited: 
Energy-Momentum Tensors, Symmetries, and the Habitat of Globally
Conserved Quantities}

\author{Domenico Giulini\\
Institute for Theoretical Physics\\
Riemann Center for Geometry and Physics\\
Leibniz University of Hannover, Germany\\
\emph{giulini@itp.uni-hannover.de}}
\date{}

\maketitle
\vspace{-0.5cm}
\begin{abstract}
\noindent
The energy-momentum tensor for a particular matter component 
summarises its local energy-momentum distribution in terms of 
densities and current densities. We re-investigate under what 
conditions these local distributions can be integrated to 
meaningful global quantities. This leads us directly to 
a classic theorem by Max von Laue concerning integrals of 
components of the energy-momentum tensor, whose statement and 
proof we recall. In the first half of this paper we do this 
within the realm of Special Relativity and in the traditional 
mathematical language using components with respect 
to affine charts, thereby focusing on the intended physical 
content and interpretation. In the second half we show how 
to do all this in a proper differential-geometric fashion and 
on arbitrary space-time manifolds, this time focusing on the 
group-theoretic and geometric hypotheses underlying these 
results. Based on this we give a proper geometric statement 
and proof of Laue's theorem, which is shown to generalise 
from Minkowski space (which has the maximal 
number of isometries) to space-times with significantly less
symmetries. This result, which seems to be new, not only 
generalises but also clarifies the geometric content and 
hypotheses of Laue's theorem. A series of three appendices 
lists our conventions and notation and summarises some of the
conceptual and mathematical background needed in the main text. 
\medskip 

\emph{Mathematics Subject Classification (MSC):} 
53D20, 58D19, 83A05\\
\emph{Physics and Astronomy Classification Scheme(PACS)}:
03.30.+p, 11.30.Cp, 11.30.-j

\smallskip
\noindent
\emph{Acknowledgements:}
I thank the Collaborative Research Centre (SFB) 1227 \emph{DQ-mat}
of the DFG at Leibniz  University of Hannover for support.
\end{abstract}

\newpage
\setcounter{tocdepth}{2}
\tableofcontents
\newpage

\section*{Introduction}
\label{sec:Introduction}
This paper deals with a well known problem in classical 
relativistic field-theory concerning the question of 
precisely how a local distribution of energy and momentum 
can be used to construct global quantities that, with some 
justification, can be said to represent these quantities 
globally.

The main body of the paper is organised into two Sections 
of approximately the same length and six subsections each.
 In addition there are three fairly detailed appendices
 providing some standard background material that we
 thought should not clutter the main text but the content
 of which is nevertheless important in order to follow 
the conceptual and technical arguments of the two main
 sections.

In Section~\ref{sec:TraditionalPicture} we will review and 
critically discuss the necessary and sufficient conditions 
under which meaningful global quantities may be constructed
for Poincar\'e invariant field theories, that is, in the 
context of Special Relativity (henceforth abbreviated SR).
We will put particular emphasis on the requirement that 
the global quantities in question must change according 
to a particular representation of the Poincar\'e group 
(inhomogeneous Lorentz group) under the action of that 
group on all those fields whose energy-momentum
 distribution is considered. This is a non-trivial
 requirement, which is conceptually essential for the
 proper physical interpretation of the quantities to be
 constructed. We will find that it takes the simple
 mathematical form of a condition of equivariance, 
provided a careful distinction is made between the
 dynamical fields (carrying energy and momentum) and
 auxiliary ones. This will then naturally lead us to 
a classic theorem dating back to 1911 due to Max von 
Laue (still without the nobiliary particle ``von'' in 
his name at the time he wrote the relevant paper,
hence it is usually referred to simply as ``Laue's
 theorem''), the proof and physical relevance of which 
we will also discuss. 

The mathematical language of
 Section~\ref{sec:TraditionalPicture} 
deliberately sticks to the traditional component
 representation of geometric objects over Minkowski 
space. It is this component language that is most 
familiar to working physicists and that is often the 
most direct way to import the physical intuition and 
intentions behind various constructions. Most of what is 
presented in Section~\ref{sec:TraditionalPicture} is 
well known, in one form or another, though I believe 
that some of the subtle points are here spelled out 
more clearly and comprehensibly than in other sources.

In Section~\ref{sec:GeometricTheory} we will then 
transform the preceding, coordinate-based discussion 
into a proper differential-geometric language. As is 
widely appreciated, coordinate-based language not only
tends to hide the geometric assumptions that are necessary
in order to render the whole construction meaningful, but
it also bears the danger to mislead to geometrically
meaningless constructions, like, e.g., taking integrals
 over tensor components -- a construction often encountered 
in physics. This is precisely what is done in the
 formulation and proof of Laue's theorem, of which no 
proof in component-free form has hitherto been given, nor
 have its geometric hypotheses been spelled out
correctly. We will fill that gap in
Section~\ref{sec:GeometricTheory}
in which we give a proper statement and proof of a 
generalisation of Laue's theorem that remains valid 
outside SR and wich contains the classical statement
as a special case.  

Throughout, and in particular in 
Section~\ref{sec:GeometricTheory}, we shall make free use
 of standard differential geometric concepts, like
 manifolds, tensor bundles and their sections, natural
 bundles, various forms of derivatives (exterior, Lie,
 covariant, exterior-covariant), Hodge duality, etc. All
 these should be familiar to physicists used to geometric
 field theory, particularly those familiar with General
 Relativity (henceforth abbreviated by GR) on the level 
of a good modern text-book, like \cite{Straumann:GR2013}.
Our notation and conventions relating to these concepts 
are explained in detail in Appendices~B and C. 
Appendix~\ref{sec:Appendix-MinkowskiSpace} contains some 
background material relating to Minkowski space and 
its automorphism group, the Poincar\'e group. 
 Appendix~\ref{sec:Appendix-MinkowskiSpace} is also meant to 
remind the reader on the notion of affine spaces and
 associated structural elements, like affine bases, 
affine charts, and affine groups, as well as the
 distinction between active and passive transformations.
 Albeit basic, these notions often tend to be 
suppressed or even forgotten in the physics literature. 
Being aware of them is important in order to fully 
appreciate the discussion in 
Section~\ref{sec:TraditionalPicture}, to which we 
now turn.

\section{The traditional picture}
\label{sec:TraditionalPicture}
In this section, which contains six subsections, we start 
with our first half of the program as just outlined. As 
already stressed, we deliberately adopt a traditional 
mathematical language, usually employed in physics papers 
and physics textbooks, to keep as close as possible 
contact to the intentions and intuitions coming from 
physics. At the same time we try to carefully avoid
some of the misconceptions which are often imported by 
either assuming some unwarranted mathematical properties,
or by misinterpreting them. A typical example of the 
latter sort results from an insufficient distinction  
between the two meanings of Poincar\'e transformations
(compare Appendix~\ref{sec:Appendix-PassiveActive}):
First, as mere transformations between affine charts on 
Minkowski space $M$ (passive interpretation), which 
necessarily affect the component representation of 
\emph{all} geometric structures defined over $M$, and, 
second, automorphism of Minkowski space (i.e. special 
affine maps), which we can use to act selectively on only 
\emph{some} geometric structures which are physically 
distinguished in a contextual fashion. For example, we 
may wish to act with a boost transformation on that subset 
of fields in Minkowski space whose energy-momentum 
distribution we consider, while leaving all other fields
in their original state. That is, we may wish to boost 
\emph{some} of the matter fields \emph{relative} to 
all the others providing the reference. This will be the 
situation we actually encounter below.  

\subsection{Energy-momentum tensors}
\label{sec:EnergyMomentumTensors}
The energy-momentum distribution of a physical system is
characterised by its \emph{energy-momentum tensor} $\EMT$,
which is a section in the bundle $TM\otimes TM$. Following 
standard conventions we write 
$\EMT\in\Section(TM\otimes TM)$(compare 
Appendix~\ref{sec:Appendix-ConventionsNotation}). So far 
$(M,g)$ may represent any spacetime. Here we are
 restricting attention to four-dimensional $M$ (this will
be relaxed in the second section). Following standard 
conventions in physics we denote space-time indices 
by greek letters if they refer to four dimensions
 (otherwiswe latin) and let 
$\alpha,\beta,..,\in\{0,1,2,3\}$, where the index 
$0$ refers to a timelike basis element or coordinate.

Consider a point $p\in M$ and a timelike 
future-pointing vector $u\in T_pM$ normalised to $c$
(velocity of light), so that $g(u,u)=c^2$ . That 
vector defines an instantaneous state of motion of an
observer at $p$ whose four-velocity is just $u$.
We take $e_0:=u/c$ as a first (timelike) vector of 
a orthonormal basis $\{e_0,e_1,e_2,e_3\}$ of $T_pM$, such
 that $g(e_\alpha,e_\beta)=g_{\alpha\beta}=\mathrm{diag}
(1,-1,-1,-1)$. 
We call $u_\perp:=\{v\in T_pM:g_p(u,v)=0\}=\mathrm{span}\{e_1,e_2,e_3\}$
the instantaneous rest-space of the observer $u$. 

If $M$ is Minkowski space (an affine space; see 
Appendix~\ref{sec:Appendix-MinkowskiSpace}) we may identify each 
$T_pM$ with the underlying vector space $V$. Then the
 subset $p+u_\perp$ of $M$ is a spacelike affine 
hyperplane consisting of all events in $M$ which are
 Einstein synchronous with $p$ for an inertial (i.e. 
force free) observer moving along the timelike line 
tangent to $u$ through $p$. All this is standard terminology 
in physics and mathematical relativity; see, e.g., 
\cite{Sachs.Wu:GR-Math} for a good account. 

Let $\{\theta^0,\theta^1,\theta^2,\theta^3\}$ 
be the dual basis to $\{e_0,e_1,e_2,e_3\}$, i.e. 
$\theta^\alpha(e_\beta)=\delta^\alpha_\beta$. Then 
$T^{\alpha\beta}=\EMT(\theta^\alpha,\theta^\beta)$
are the so-called contravariant components of 
$\EMT$, whose physical interpretation is as follows:
$T^{00}$ is the energy density, $\vec S:=(cT^{01},cT^{02},cT^{03})$ 
are the three components of the energy current-density,  
$\vec G:=(T^{10}/c,T^{20}/c,T^{30}/c)$ are the three components 
of the momentum density, and finally the symmetric 
$3\times 3$ matrix $\{T^{ab}\}$ comprises the $9$ components 
of the momentum current-density. Here the word ``components'' 
always refers to vectors and tensors in the instantaneous 
3-dimensional rest frame $u_\perp$ with respect to the basis 
$\mathrm{span}\{e_1,e_2,e_3\}$. As such $\EMT$ comprises  
$1+3+3+9=16$ independent components characterising the 
energy-momentum distribution of matter at each spacetime point.   

Often, though not always, $\EMT$ is assumed to be symmetric,
\begin{equation}
\label{eq:EM-Symmetry}
T^{\alpha\beta}=T^{\beta\alpha}\,.
\end{equation}  
In the notation above this means that the energy current-density 
$\vec S$ and the momentum density $\vec G$ are related by 
$\vec G=\vec S/c^2$ and that the $n$-th component of the $m$-th
momentum current-density equals the $m$-th component of the 
$n$-th momentum current-density, i.e. $\Sigma^{mn}=\Sigma^{nm}$. 
Hence a symmetry \eqref{eq:EM-Symmetry} reduces the $16$ 
to $1+3+6=10$ independent components. Energy momentum 
tensors representing all the sources of gravitational 
fields in Einstein's field equations of GR are necessarily
symmetric, though this does not necessarily imply that 
they cannot be meaningfully decomposed into a sum of non
symmetric ones attributed to various sub-components
of matter. In any case, we shall assume $\EMT$ to be 
symmetric throughout this paper. 

There is another condition that $\EMT$ may satisfy and 
which we wish to mention here, although we shall not 
need it right now (it becomes of central importance later 
on). It is the covariant divergencelessness of 
$\EMT$, which we write as $\nabla\cdot\EMT=0$, or in
components 
\begin{equation}
\label{eq:EM-Conservation}
\nabla_\beta T^{\alpha\beta}=0\,.
\end{equation}  
Here $\nabla$ denotes the Levi-Civita covariant derivative for $g$,
that is, the unique torsion-free and metric (i.e $\nabla g=0$)
connection (we use ``covariant derivative'' and ``connection'' 
synonymously).  
Equation \eqref{eq:EM-Conservation} implies the local
 conservation of energy-momentum encoded by $\EMT$,
 provided $g$ admits certain symmetries generated by
 Killing vector fields. Without such symmetries energy 
and momentum of the matter alone will not be conserved. 
In that case the interpretation of 
\eqref{eq:EM-Conservation} in GR is that locally energy 
and momentum can be exchanged between the matter 
(described by $\EMT$) and the gravitational field 
(described by $g$); see, e.g.,
\cite{Straumann:GR2013}.

\subsection{Global energy momentum in SR}
\label{sec:GlobalEM-inSR}
Let us now restrict $(M,g)$ to Minkowski space. 
In manifold terms this means that $M$ is diffeomorphic 
to $\reals^4$ and that $g=\eta$, where here and in 
the sequel $\eta$ denotes a Lorentz metric whose 
Levi-Civita connection is flat. We also recall the 
affine structure of Minkowski space and refer to 
Appendix~\ref{sec:Appendix-MinkowskiSpace} for details. Important here 
is that the affine structure endows the manifold $M$ with a
 preferred set of global charts, called \emph{affine charts}.
With respect to affine charts all connection coefficients 
vanish and \eqref{eq:EM-Conservation} assumes the form 
\begin{equation}
\label{eq:EM-ConservationSRT}
\partial_\beta T^{\alpha\beta}=0\,,
\end{equation}   
which can indeed be interpreted as local energy-momentum 
conservation. 

Let us digress for a moment in order to make a 
few comment concerning the physical significance 
of the affine structure of Minkowski space. 
As is well known, the timelike lines within the 
set of all lines (one-dimensional affine subspaces)
 already fix the entire affine structure of 
Minkowski space. 
In this fashion the affine structure of Minkowski 
space becomes a faithful mathematical representation 
of the physical law of inertia, according to which 
a particular subset of motions is characterised as
 ``force-free'' or, as already said, ``inertial''. 
In this way affine maps, too,  acquire  the physical
 significance as the symmetry group of the law of
 inertia:  they map force-free motions to  
force-free motions. In that sense the law of 
inertia should be regarded as a special case of 
a ``path structure'' in the sense of
 \cite{Ehlers.Koehler:1977}, without the 
specification of which the very notion 
of ``force'' makes no sense. By definition,  
``forces'' are the causes for deviations from 
inertial motions, which must be defined first. 
For more on this, see \cite{Giulini:2002b}
and \cite{Pfister.King:InertiaAndGravitation}.
The take-home message for us at this point is
that in using affine structures we are, in fact, 
already in the process of modelling physical 
dynamical principles. 

After this small digression we return to our 
main discussion. We note that Hypersurfaces of 
constant value for a single affine-chart
 coordinates are affine hyperplanes. In such charts a 
moment in time is represented by a spacelike affine 
hyperplane $\Sigma\subset M$ in Minkowski space; e.g., 
the hyperplane of constant time zero: 
$\Sigma:=\{p\in M:x^0(p)=0\}$. The global quantity 
we wish to construct is either associated with the whole 
hyperplane (i.e. involves all of space at this time) or a 
bounded subset thereof, which we shall also denote by 
$\Sigma$. For its construction we do not wish to impose
\eqref{eq:EM-ConservationSRT} in order to be able to speak of 
the system's overall energy and momentum in situations 
where it is not conserved. Hence, in what follows next, 
we shall only assume  symmetry \eqref{eq:EM-Symmetry}. 
Later we will have plenty of opportunity to discuss the 
impact of \eqref{eq:EM-ConservationSRT}

 We wish to speak of the \emph{total} energy and momentum 
contained in $\Sigma$ (i.e. ``at the moment in time''
represented by $\Sigma$). The most naive thing 
to do is to just integrate the energy density $T^{00}$
and momentum density\footnote{Actually, as explained above, 
the momentum density is $1/c$ times the components 
$T^{a0}$ ($a=1,2,3$), which means that our $P^\alpha$
will actually equal $c$ times the components of 
ordinary four-momentum. In order 
not to carry along all the factors of $c$  we shall ignore 
this difference, which is completely irrelevant for our 
purpose and which disappears in units where $c=1$.} 
$T^{a0}$ ($a=1,2,3$) over $\Sigma$. Hence we consider the 
four numbers, $\{P^\alpha: \alpha=0,1,2,3\}$ defined by 
\begin{equation}
\label{eq:SRT-FourMomentum}
P^\alpha:=\int_\Sigma T^{\alpha\beta}n_\beta\,d\mu_\Sigma\,.
\end{equation}
Here $n_\beta$ are the covariant components of the normal 
to $\Sigma$ in $M$ and $d\mu_\Sigma$ is the measure 3-form 
on $\Sigma$ induced from the measure 4-form $\varepsilon$ 
on $(M,g)$, which in affine coordinates $x^\alpha$ is 
just the Lebesgue measure: 
$\varepsilon=dx^0\wedge dx^1\wedge dx^2\wedge dx^3$. 
Writing $i_n$ for the insertion map that inserts $n$ 
into the first slot of a general form, we have
\begin{equation}
\label{eq:SRTInducedMeasure}
d\mu_\Sigma=\star n^\flat=i_n\varepsilon=
\tfrac{1}{3!}n^\mu \varepsilon_{\mu\alpha\beta\gamma}
dx^\alpha\wedge dx^\beta\wedge dx^\gamma\,.
\end{equation}
Here and in the following the superscript $\flat$ on
a vector denotes its corresponding dual vector under 
the isomorphism given by the metric, which we define 
in \eqref{eq:DefMusicalIsomorphisms}. Also $\star$ 
denotes the Hodge duality map defined in
\eqref{eq:DefHodgeDuality}.
 
The intended result of the somewhat bold method to just 
integrate all the densities is, of course, that these 
four numbers somehow represent the total energy and total
 translatory momentum of the system in the rest-frame 
represented by the hyperplane $\Sigma$ and with respect 
to the affine coordinates $x^a$ used within in. We further 
expect these four numbers to define an element in a vector 
space that carries a specific representation of the Lorentz 
group, namely the defining one.  Physicists often express 
this expectation by saying that total energy-momentum should
form a ``four vector'', meaning that the four numbers
 transform according to the defining representation of 
$\mathrm{SO}(1,3)$. But is that really true? And if so, 
how would we characterise the vector space this 
``energy-momentum vector'' is an element of? In fact, 
as is well known, and as we will see in a moment, the 
fulfilment of all these expectations is not at all 
guaranteed and depends on further global conditions 
on $\EMT$ that we will now investigate. 

\subsection{Local energy-momentum transformation}
\label{sec:LocalEnergyMomentumTransformation}
The four numbers \eqref{eq:SRT-FourMomentum} we now calculate 
for two different dynamical systems whose energy-momentum 
distributions are represented by two different energy-momentum 
tensors, $\EMT$ and $\bar\EMT$. The two systems are chosen such 
that they are related by an active Lorentz transformation 
$\Lambda$, which means that 
\begin{equation}
\label{eq:SRTBoostedEMT-1}
(\Lambda,\EMT)\mapsto \Lambda\cdot\EMT:=\Lambda\otimes\Lambda (\EMT\circ \Lambda^{-1})=:\bar\EMT\,,
\end{equation}
where in the last step we simply introduced the notational 
abbreviation of writing an overbar for the $\Lambda$-transformed 
quantity. Note that \eqref{eq:SRTBoostedEMT-1} just denotes the 
push-forward transformation law of a contravariant tensor field 
under the action of the Lorentz transformation $\Lambda$. It is 
a special example for the natural lift of the diffeomorphic action 
of the Lorentz group on $M$ to any bundle associated to the 
principal bundle of linear frames, here to the bundle 
$TM\otimes TM$ (compare Appendix~\ref{sec:Appendix-NaturalBundles}).
In terms of component functions with respect to our affine 
coordinates, this reads
\begin{equation}
\label{eq:SRTBoostedEMT-2}
{\bar T}^{\alpha\beta}\bigl(x\bigr)
=\Lambda^\alpha_\gamma\Lambda^\beta_\delta T^{\gamma\delta}
\bigl(\Lambda^{-1}x\bigr)\,,
\end{equation}
where the $x$ and $\Lambda^{-1}x$ in the argument stands 
for all four $x^\mu$ and 
$(\Lambda^{-1})^\mu_\nu x^\nu$, $\mu\in\{0,1,2,3\}$,
respectively. 

The following is important to note: For \emph{both} dynamical 
systems the components of their respective dynamical quantities 
refer to the \emph{same} affine coordinates 
$\{x^\alpha:\alpha=0,1,2,3\}$. Our Poincar\'e  transformation 
acts on the dynamical fields, not on the observer or a 
reference-system, and produces via its action another system 
with different observable features relative to the very same 
reference system that we exclusively use. 

To end this subsection we wish to alert the reader to another 
non-trivial aspect hidden in \eqref{eq:SRTBoostedEMT-1}, also
connection with the active interpretation of transformations
employed here, which is the following: Usually we think of 
$\EMT$ not as just simply given, but as functionally 
specified in terms of elementary fields, which we separate 
into two sets, $F$ and $F'$, whose distinction is important.
 Hence we write $\EMT(F,F')$. $F$ collectively denotes all 
fields of which $\EMT$ represents the energy-momentum
distribution, and $F'$ collectively denotes all possibly 
existent further background fields that we may use to 
define $\EMT$. For example, in Maxwellian electrodynamics
in its traditional ``metric'' formulation\footnote{The 
metric formulation of electrodynamics assumes a background 
metric of spacetime, the Minkowski metric. There are also 
metric-free formulations, as discussed in detail in \cite{Hehl.Obukhov:Electrodynamics}.},
$F$ would be the Faraday tensor, comprising electric and 
magnetic fields, and $F'$ would be the background 
Minkowskian metric. The value-space of the dynamical 
fields $F$ is assumed to carry some representation 
$D$ of the Lorentz group, so that the field $F$ itself (as element 
of a mapping space; later to be identified with a section in an 
appropriate bundle) transforms as 
\begin{equation}
\label{eq:SRT-MappingOfField}
(\Lambda,F)\mapsto \Lambda\cdot F:=D(\Lambda)\circ F\circ\Lambda^{-1}\,.
\end{equation}
The non-trivial requirement behind \eqref{eq:SRTBoostedEMT-1} 
then is:
\begin{equation}
\label{eq:SRTBoostedEMT-3}
\EMT(\Lambda\cdot F\,,\,F')=\Lambda\cdot\EMT(F,F')\,.
\end{equation}
Note that the corresponding statement 
\begin{equation}
\label{eq:SRTBoostedEMT-4}
\EMT(\Lambda\cdot F\,,\,\Lambda\cdot F')=\Lambda\cdot\EMT(F,F')\,,
\end{equation}
in which $\Lambda$ is also allowed to act on $F'$ 
as $(\Lambda,F')\mapsto \Lambda\cdot F':=D'(\Lambda)\circ F'\circ\Lambda^{-1}$ via some representation $D'$ on
the value-space of $F'$, would be much weaker and, in fact, 
almost trivial. It would merely require $\EMT$ to be a 
covariant construct of $F$ and $F'$. In contrast, equation 
\eqref{eq:SRTBoostedEMT-3} is far more demanding. 
It requires that the energy-momentum distribution of the 
$\Lambda$-shifted fields $F$ is the $\Lambda$-shift 
of the original energy-momentum distribution of $F$,
given the same background $F'$ throughout. For example,
coming back to ordinary (metric) electromagnetism in
Minkowski space, let now $\Lambda$ be any diffeomorphism, 
$F$ again the Farady tensor, and $F'$ the background 
Minkowski metric. Then \eqref{eq:SRTBoostedEMT-4} will 
always hold true, but \eqref{eq:SRTBoostedEMT-3} only 
if $\Lambda$ stabilizes $F'$, i.e. $\Lambda\cdot F'=F'$, 
which just says that $\Lambda$ must be an isometry and 
hence a Poincar\'e transformation.  

It seems remarkable that the non-trivial requirement 
behind  \eqref{eq:SRTBoostedEMT-3} is hardly ever spelled 
out, or even recognised, in standard text-books. 
A notable exception is found in \S\,31 and \S\,$31^*$ of Fock's 
classic text \cite{Fock:STG}, where Fock discusses the 
requirement (which he calls a ``physical principle'')
that ``the mass tensor must be a function of the state 
of the system'' (\cite{Fock:STG}, p.\,95). Here ``mass 
tensor'' refers to the energy-momentum tensor rescaled by 
$c^{-1}$. This sound as if Fock had \eqref{eq:SRTBoostedEMT-3}
in mind, which is also supported by his discussion that 
stresses the non-triviality of his ``physical principle''. 
However, semantically an ambiguity remains as to whether 
Fock's ``system'' refers to, in our notation, the 
system represented by $F$, or the wider system represented 
by $(F,F')$. Fock unfortunately does not explicitly address 
this distinction which becomes particularly important in 
the more general cases discussed below 
in\,\ref{sec:GroupActionsAndCharges}.

\subsection{Global energy-momentum transformation}
\label{sec:GlobalEM-Trans}
We now wish to see how the four numbers 
\eqref{eq:SRT-FourMomentum} are affected by a 
particular set of Lorentz transformations, namely 
pure boosts.\footnote{Pure boosts can be characterised 
within the Lorentz group only with reference to a 
timelike direction, which here is taken to be 
$\partial/\partial x^0$. A `pure boost' is then 
any linear Lorentz transformation that moves points 
in a timelike 2-plane containing $\partial/\partial x^0$
and leaves the $g$-orthogonal spacelike 2-plane poinwise 
fixed.Pure boosts do not form a subgroup of the Lorentz
group.} We restrict to boosts because no non-trivial 
statement of the kind developed here follows from translations
and rotations. 

So let us specifically take a boost in $x^1$-direction:
\begin{equation}
\label{eq:SRTPureBoost}
\Lambda^\alpha_\beta=
\begin{pmatrix}
\gamma&\beta\gamma&0&0\\
\beta\gamma&\gamma&0&0\\
0&0&1&0\\
0&0&0&1\\
\end{pmatrix}\,.
\end{equation}
Here, as usual,  $\beta:=v/c$ and $\gamma:=1/\sqrt{1-\beta^2}$. Then
\begin{subequations}
\label{eq:SRT-EMT-Boost-Trans}
\begin{alignat}{2}
\label{eq:SRT-EMT-Boost-Trans-a}
&{\bar T}^{00}(x)
&&=\gamma^2\Bigl[
T^{00}(\underline{x})
+\beta^2 T^{11}(\underline{x})
+2\beta T^{01}(\underline{x})\Bigr]\\
\label{eq:SRT-EMT-Boost-Trans-b}
&{\bar T}^{11}(x)
&&=\gamma^2\Bigl[
T^{11}(\underline{x})
+\beta^2 T^{00}(\underline{x})
+2\beta T^{01}(\underline{x})\Bigr]\\
\label{eq:SRT-EMT-Boost-Trans-c}
&{\bar T}^{01}(x)
&&=\gamma^2\Bigl[
\bigl(1+\beta^2)T^{01}(\underline{x}\bigr)
+\beta\bigl(
T^{00}(\underline{x})+T^{11}(\underline{x})
\bigr)\Bigr]\\
\label{eq:SRT-EMT-Boost-Trans-d}
&{\bar T}^{0n}(x)
&&=\hspace{3pt}\gamma\,\Bigl[
T^{0n}(\underline{x})+\beta\,T^{1n}(\underline{x})\Bigr]\\
\label{eq:SRT-EMT-Boost-Trans-e}
&{\bar T}^{1n}(x)
&&=\hspace{3pt}\gamma\,\Bigl[
T^{1n}(\underline{x})+\beta\,T^{0n}(\underline{x})\Bigr]\\
\label{eq:SRT-EMT-Boost-Trans-f}
&{\bar T}^{nm}(x)
&&=\,T^{nm}(\underline{x})\,.
\end{alignat}
Here we set for abbreviation,
\begin{equation}
\label{eq:SRT-EMT-Boost-Trans-g}
\underline{x}:=\Lambda^{-1}x=\bigl(
\gamma(x^0-\beta x^1)\,,\
\gamma(x^1-\beta x^0)\,,\
x^2\,,\
x^3\bigr)\,,
\end{equation}
\end{subequations}
and the indices $n$ and $m$ may take on independently any
of the two values $2$ and $3$.

In the integral \eqref{eq:SRT-FourMomentum} we have $n^\mu=1$
for $\mu=0$ and $n^\mu=0$ otherwise, for $\Sigma$ is the zero-level
set of $x^0$. This implies that the induced measure 
\eqref{eq:SRTInducedMeasure} is simply the Lebesgue measure in 
the spatial affine coordinates
\begin{equation}
\label{eq:SRT-LebesgueMeasure}
d\mu_\Sigma =dx^1\wedge dx^2\wedge dx^3=:d^3x\,.
\end{equation}
For ${\bar P}^\alpha$ the integrals \eqref{eq:SRT-FourMomentum} now involve the four 
components ${\bar T}^{00}$, ${\bar T}^{01}$ and ${\bar T}^{0n}$
($n=2,3$), evaluated on $\Sigma$. For $x^0=0$ we have 
$\underline{x}
=\bigl(
-\beta\gamma x^1\,,\,
\gamma x^1\,,\,
x^2\,,\,
x^3
\bigr)$. Therefore the integrands \eqref{eq:SRT-FourMomentum}
will involve the components ${\bar T}^{00}$, ${\bar T}^{01}$
and ${\bar T}^{0n}$ at all times $-\beta\gamma x^1$, 
where $x^1$ takes all values from within the intersection 
of $\Sigma$ with the support of $\EMT$. As a result, no 
simple statement relating the integrals of the original 
and the boosted systems can be made for generally time dependent $T^{\alpha\beta}$. On the other hand, if we assume 
the components $T^{\alpha\beta}$ to be independent of time $x^0$,
i.e.,
\begin{equation}
\label{eq:SRT-Stationarity}
\partial_0T^{\alpha\beta}=0\,,
\end{equation}
all dependencies of the barred components  
${\bar T}^{\alpha\beta}$ on the spatial coordinates 
derive from the spatial dependence of $T^{\alpha\beta}$
in the simple fashion given by 
\eqref{eq:SRT-EMT-Boost-Trans-g} for $x^0=0$: 
$x^1$ is multliplied with $\gamma$ whereas the 
dependence on $x^2$ and $x^3$ is unchanged. By a simple 
change of variables in the integration over $\Sigma$ 
this results in a common factor $1/\gamma$ and we obtain  
\begin{subequations}
\label{eq:SRT-FourMomentumTrans}
\begin{alignat}{2}
\label{eq:SRT-FourMomentumTrans-a}
&{\bar P}^0&&\,=\,\gamma
\left(
P^0+2\beta P^1+\int_\Sigma T^{11}d^3x 
\right)\,,\\
\label{eq:SRT-FourMomentumTrans-b}
&{\bar P}^1&&\,=\,\gamma
\left(
(1+\beta^2)P^1+\beta P^0+\beta\int_\Sigma T^{11}d^3x  
\right)\,,\\
\label{eq:SRT-FourMomentumTrans-c}
&{\bar P}^n&&\,=\,
P^n+\beta \int_\Sigma T^{1n}d^3x\,. 
\end{alignat}
\end{subequations}
This we compare to the defining representation 
of the Lorentz group on $\mathbb{R}^4$: 
\begin{subequations}
\label{eq:SRT-FourVectorTrans}
\begin{alignat}{2}
\label{eq:SRT-FourVectorTrans-a}
&{\bar X}^0&&\,=\,\gamma
\left(
X^0+\beta X^1\right)\,,\\
\label{eq:SRT-FourVectorTrans-b}
&{\bar X}^1&&\,=\,\gamma
\left(
X^1+\beta X^0\right)\,,\\
\label{eq:SRT-FourVectorTrans-c}
&{\bar X}^n&&\,=\,X^n \,.
\end{alignat}
\end{subequations}
It follows that the $P^\alpha\rightarrow{\bar P}^\alpha$
transformation follows the pattern \eqref{eq:SRT-FourVectorTrans} 
for all $\beta$, if and only if $P^1=0$ and all three
integrals $\int_\Sigma T^{1a}d^3x$ vanish for $a=1,2,3$.

Now, as the direction of our boost was arbitrary, 
repeating the argument with boosts in the 
$2-$ and $3-$ direction shows that
$P^\alpha\rightarrow{\bar P}^\alpha$ transforms like
the components of four vector under boosts, if and only 
if the integrals over $\Sigma$ of \emph{all} components 
$T^{\mu\nu}$ vanish, except that of $T^{00}$. Hence
\begin{equation}
\label{eq:SRT-VanishingIntegrals}
\int_\Sigma T^{\mu m}\,d^3x=0\
\end{equation}
for all $\mu\in\{0,1,2,3\}$ and all $m\in\{1,2,3\}$.

A priori there seems to be no obvious reason why  
any stationary physical system should have an 
energy-momentum tensor satisfying \eqref{eq:SRT-VanishingIntegrals},
In fact, \eqref{eq:SRT-VanishingIntegrals} will fail in the general 
stationary case. But recall that so far we did not assume 
energy-momentum conservation \eqref{eq:EM-ConservationSRT}. 
It is Laue's theorem that precisely connects energy-momentum 
conservation and \eqref{eq:SRT-VanishingIntegrals}.

\subsection{Laue's theorem: classical statement}
\label{sec:LaueClassical} 
Laue's theorem was first proven (under slightly stronger 
hypotheses, some of which turn out to be unnecessary) 
in his classic paper \cite{Laue:1911b}, in which he for 
the first time explained in a general fashion the impact 
of Special Relativity onto the dynamical description of 
systems in static equilibrium. The theorem is based on 
the obvious identity, 
\begin{equation}
\label{eq:SRT-LaueThmProof}
\partial_n(T^{\mu n}\,\varphi)=(\partial_n T^{\mu n})\,\varphi+T^{\mu n}\,\partial_n\varphi\,,
\end{equation}
where $\varphi:\Sigma\rightarrow\reals$ is any 
smooth function. 
If we assume energy-momentum conservation 
\eqref{eq:EM-ConservationSRT} and stationarity 
\eqref{eq:SRT-Stationarity}, the first 
term on the right-hand side vanishes. 
Applying Gau{\ss}' theorem (Stokes' theorem 
in 3-dimensions), we get 
\begin{equation}
\label{eq:SRT-GaussTheoremSpatial-1}
\int_\Sigma T^{\mu n}\partial_n\varphi\,d^3x=
\int_{\partial\Sigma} T^{\mu n} \varphi \nu_n\,do\,,
\end{equation}
where $\nu_n$ are covariant components of the 
outward-pointing normal of the boundary 
$\partial\Sigma$ of $\Sigma$ and $do$ is the 
induced measure on the 2-dimensional boundary.
Now, if $(T^{\mu n} \varphi)\vert_{\partial\Sigma}=0$,
e.g., for all $\varphi$ with compact support in the
interior of $\Sigma$, we have
\begin{equation}
\label{eq:SRT-GaussTheoremSpatial-2}
\int_\Sigma T^{\mu n}\partial_n\varphi\,d^3x=0\,.
\end{equation}
If we forget about the index $\mu$ for a moment, 
this is just the well-known statement that a
divergence-free vector field is 
$L^2(\reals^3,d^3x)$-orthogonal to all gradient fields (given sufficient fall-off
conditions in the case of non-compact support). 

Now, for the special case $\varphi=x^m$ this reduces
to \eqref{eq:SRT-VanishingIntegrals}, provided 
that the boundary integral in 
\eqref{eq:SRT-GaussTheoremSpatial-1} still vanishes. 
This poses further conditions on $\EMT$ regarding 
its behaviour near $\partial\Sigma$. 
An obvious sufficient condition would be that the 
intersection of $\EMT$'s support
\begin{equation}
\mathrm{supp}(\EMT):=
\overline{\bigl\{p\in M : \EMT(p)\ne 0\bigr\}}
\end{equation}
with $\Sigma$ is compact. On the other hand, if 
$\mathrm{supp}(\EMT)\cap\Sigma$ is not compact, 
it would be sufficient that $\EMT$ has a $1/r^{3+\epsilon}$ fall-off (i.e. faster than $1/r^3$; 
$\epsilon>0$) at each end\footnote{``Ends'' of 
manifolds are connected subsets 
not contained in any compact connected subset. 
This notion was coined in \cite{Freudenthal:1931}. 
They are sometimes called ``asymptotic regions'' 
in the physics literature.} of $\Sigma$, where $r$
 denotes spatial geodesic distance. Note that 
for radiating systems $\Sigma$ needs to be a Cauchy
 surface in order for the requirement of compact 
support to make sense. The reason for this is 
explained in Fig.\,\ref{fig:ConformalMinkowskiSpace}.

\begin{figure}[ht!]
\begin{center}
\includegraphics[width=0.71\linewidth]%
{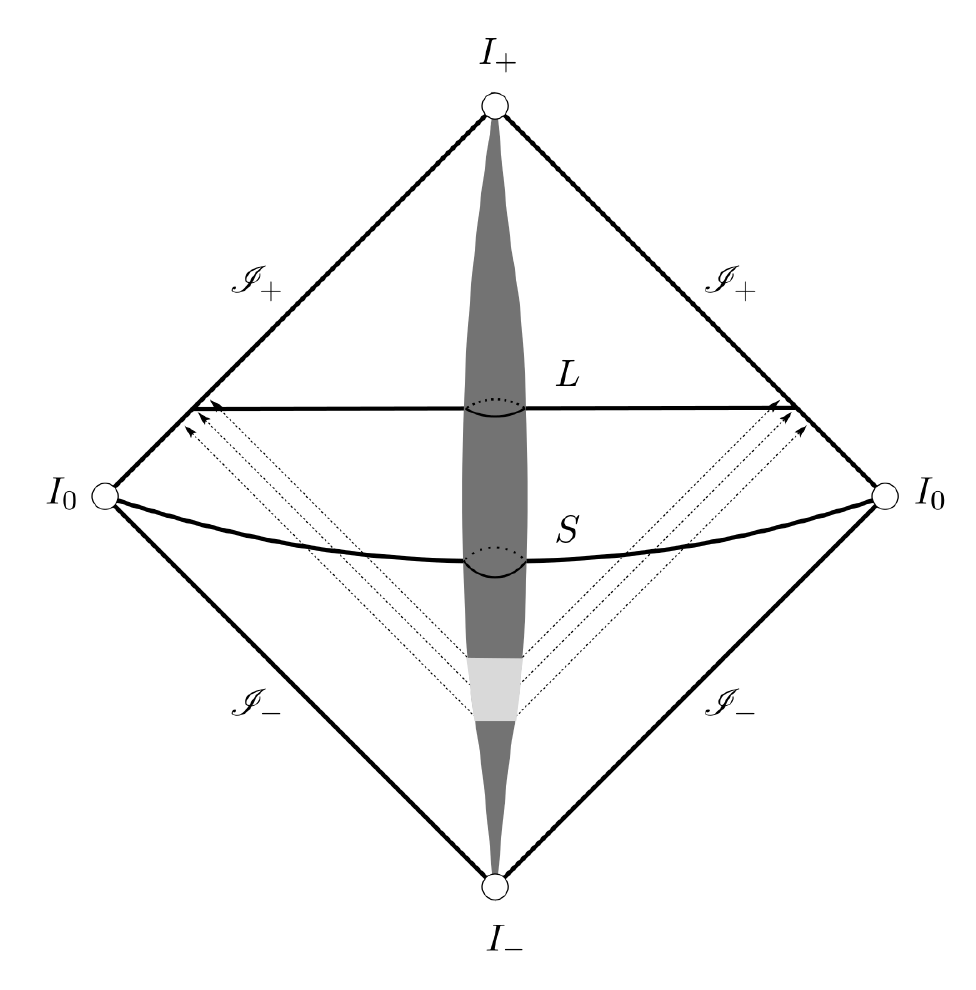}
\end{center}
\caption{\label{fig:ConformalMinkowskiSpace}
The figure shows the conformal compactification
of Minkowski space, the boundary of which decomposes into 
future/past timelike infinity $I_+/I_-$, future/past 
lightlike infinity $\mathscr{I}_+/\mathscr{I}_-$,
and spacelike infinity $I_0$. The vertical dark-shaded 
region extending from $I_-$ to $I_+$ corresponds to the 
``world tube'' of the physical system whose energy-momentum
tensor is denoted by $\EMT$ in the text. The lighter shaded 
region within that denotes the events during which the 
system radiates. The dotted straight lines at 45 degrees 
connecting that lighter shaded region to $\mathscr{I}_+$ 
indicate the lightlike rays of radiation. We have drawn 
two spacelike hypersurfaces: One denoted 
by $S$, which is asymptotically flat 
and extends all the way to $I_0$, and another one, 
$L$, which is asyptotically hyperboloidal (of constant 
negative sectional curvature) and which extends to 
$\mathscr{I}_+$. The former is a Cauchy hypersurface, i.e.
its domain of dependence is all of Minkowski space,
the other is not. For the latter it is clear that the 
radiation that the system released during the finite 
time interval (light shaded region) in the finite past 
intersects an open neighbourhood in $L$ of 
$L\cap\mathscr{I}_+$. This is not true for $S$ unless 
the system had been radiating all the way down to the 
infinite past $I_-$. This is to show that for systems 
capable of radiating off energy and momentum it is 
generally inconsistent to require 
$\mathrm{supp}(\EMT)\cap\Sigma$ to be compact for 
spacelike $\Sigma$, unless $\Sigma$ is a Cauchy 
surface.}
\end{figure}

The foregoing discussion shows the following 
\begin{theorem}[Laue's theorem, classical version]
\label{thm:LauesTheoremClassic}
Let $\EMT$ be a symmetric second-rank contravariant
tensor in Minkowski space $(M,\eta)$ -- from now on 
called the \emph{energy-momentum tensor} --, which is 
such that there exists a global affine chart
$\{x^0,x^1,x^2,x^3\}$ such that $\partial T^{\alpha\beta}/\partial x^0=0$. 
Let $\Sigma\subset M$ be the spacelike hyperplane
$x^0=0$, or a bounded subdomain thereof, which we fix 
once and for all. Consider now the four intgrals  
\eqref{eq:SRT-FourMomentum}, which we write 
$P^\alpha[\Sigma,\EMT]$, and the corresponding four 
integrals for the Lorentz-boosted energy-momentum 
tensor $\Lambda\cdot\EMT:=\Lambda\otimes\Lambda (\EMT\circ \Lambda^{-1})$. Then
\begin{equation}
\label{eq:SRT-LauesThmClassic1}
P^{\alpha}[\Sigma\,,\,\Lambda\cdot\EMT]=\Lambda^\alpha_\beta\, P^\beta[\Sigma\,,\,\EMT]
\end{equation}
for all pure boost transformations, if and only if the 
nine space integrals \eqref{eq:SRT-VanishingIntegrals}
vanish. A sufficient condition for this to happen is that  
$\EMT$ is conserved in the sense of 
\eqref{eq:EM-ConservationSRT} \emph{and} either 
$\EMT\big\vert_{\partial\Sigma}\equiv 0$ (in case 
$\partial\Sigma\ne\emptyset$) or $\EMT$ falls off like 
$1/r^{3+\epsilon}$ at each end of $\Sigma$,
where $r$ is the geodesic distance with respect to  
any reference point inside $\Sigma$.
\end{theorem}

The classic literature on special relativity is full of more or 
less convincing derivations of this result. Laues original 
derivation \cite{Laue:1911b} and all (as far as I am aware) of its followers, even the most recent \cite{Wang:2015}, replace our hypothesis of stationarity, i.e., 
$\partial T^{\alpha\beta}/\partial x^0=0$, by the stronger hypothesis of staticity, which in addition to 
stationarity also contains the three conditions 
$T^{0m}=0$. Systems satisfying this condition of 
staticity, as well as local energy-momentum
conservation \eqref{eq:EM-ConservationSRT}, are called 
\emph{complete static systems} by Laue. For them Laue 
derives the result that the integrals 
\eqref{eq:SRT-FourMomentum} form the components of 
a ``four-vector''. We shall clarify the mathematical 
content of this statement below.  

Many text-book derivations suffer from inaccuracies,
taking sufficient conditions also as necessary, or unduly 
neglecting the role of the boundary integrals. Quite 
recently, a pedagogically guided attempt has been made to 
bring some logical order into the various classic textbook 
statements of what Laue's theorem actually says, as well 
as the corresponding derivations\,\cite{Wang:2015}. 
Unfortnunately the derivation given in \cite{Wang:2015}
is itself logically incomplete\footnote{The incompleteness 
occurs in the third line of the string of equations (6)
on p.\,1471 of \cite{Wang:2015},
where the integration domain is changed from a hyperplane 
$t'=const$ to a tilted hyperplane $t=const$. This step in \cite{Wang:2015} does not follow from the initial 
hypothesis of $t$-independence, unless some form of local energy-momentum conservation is imposed in addition, 
which at this point would amount to a \emph{petitio principii}.} and entirely phrased in a non-geometric 
component language, so that it remains entirely unclear 
what becomes of Laue's theorem in curved spacetimes. 

Finally we stress once more that the requirement 
\eqref{eq:SRT-LauesThmClassic1} should not be confused 
with that in which $\Sigma$ is also acted on by 
$\Lambda$:
\begin{equation}
\label{eq:SRT-LauesThmClassic1-Fake}
P^{\alpha}[\Lambda\cdot\Sigma\,,\,
\Lambda\cdot\EMT]=\Lambda^\alpha_\beta\, P^\beta[\Sigma\,,\,\EMT]\,,
\end{equation}
where $\Lambda\cdot\Sigma$ denotes the image
of the hyperplane $\Sigma$ under the (active!) Lorentz 
transformation $\Lambda$.
This equation is trivially satisfied for all $\EMT$, 
as one easily sees by a simple change-or-variable 
transformation of the integral and the fact that 
now not only the components $T^{\alpha\beta}$ but 
also the $n_\beta$ transforms so that 
$T^{\alpha\beta}n_\beta$ already transforms as a 
four vector. Hence no further conditions on 
$\EMT$ result. Equation \eqref{eq:SRT-LauesThmClassic1-Fake} 
is a trivial requirement that merely states that 
$P^\alpha$ has been constructed from geometric 
objects. But this is not the point!
Rather, the point is whether the quantities 
$P^\alpha$ behave like a four vector under boosts 
which exclusively act on the matter fields so 
as to \emph{change} their state relative to a 
\emph{fixed} background structure.  The difference 
between \eqref{eq:SRT-LauesThmClassic1} and  
\eqref{eq:SRT-LauesThmClassic1-Fake} is easily 
overlooked if one adopts a passive interpretation 
of Lorentz transformations, in which case one is 
tempted to let $\Lambda$ act on \emph{any} geometric 
structure that appears in one's formulae (``anything that 
has indices on it is transformed''). This confusion 
has led to several fake-resolutions of the 
transformation problem in the classic literature, 
like e.g. in \cite{Rohrlich:ClassicalChargedParticles} 
and \cite{Rohrlich:TheElectron1973}.

\subsection{The many uses of Laue's theorem}
\label{sec:UsesLaueTheorem}
Before we outline the geometric theory, let us say a few words
on the many uses of Laue's theorem. Laue himself applied it to 
many of the at first sight paradoxical results in special-relativistic 
kinematics, most importantly to the null results of the classic 1903
experiment by Trouton and Noble \cite{Laue:1912}. All these apparent 
paradoxa result from applications of \eqref{eq:SRT-FourMomentumTrans}
to energy-momentum tensors violating the hypotheses of Laue's theorem,
so that the space integrals in these equations do not vanish. As an easy 
first example consider the case when the space integrals in 
\eqref{eq:SRT-FourMomentumTrans-c} (i.e. over the longitudinal-transversal 
spatial components of $\EMT$) do not vanish. This implies that as a 
result of a pure boost the systems picks up momentum perpendicular to the 
boost direction. Momentum and velocity of such a 
boosted system will no longer be aligned. This also means that putting such a 
system into the boosted state will necessitate the transmission of 
angular momentum, i.e. the action of torque. Another, historically 
earlier apparent paradox concerns the Coulomb field of charge distributions. 
Assuming spatial isotropy, the tracelessness of the electromagnetic
energy-momentum tensor ($T^{00}=T^{11}+T^{22}+T^{33}$)  gives 
$\int_\Sigma T^{11}=\frac{1}{3} P^0$. Hence, even if $P^1=0$, we have 
${\bar P}^0=\frac{4}{3}\gamma P^0$ and ${\bar P}^1=\frac{4}{3}\beta\gamma P^0$,
which looks all right except for the factor $4/3$. Systems to which this applies
are, e.g., the Coulomb field outside a spherically-symmetric surface charge
distribution (like in ``classical models of the electron''; see, e.g., \cite{Rohrlich:TheElectron1973}), or the black-body radiation inside a 
container. The reason why all these examples violate the hypotheses of
Laue's theorem is, in Laue's terminology, that they are not ``complete'' 
in the following sense: They either do not satisfy  \eqref{eq:EM-ConservationSRT} 
or, if the points of $\Sigma$ where \eqref{eq:EM-ConservationSRT} is violated 
are removed from $\Sigma$, $\EMT\big\vert_{\partial\Sigma}\ne 0$, i.e. the
boundary conditions are not met. 
For example, including the stresses that are needed to prevent the charge distribution of
the ``classical electron'' (so-called Poincar\'e stresses, because in
\cite{Poincare:1906} Poincar\'e already pointed out the necessity of such 
- possibly non-electromagnetic - stresses) from exploding, or the black-body radiation from escaping the cavity, removes all these paradoxa in a general and 
model-independent fashion. Therein, as well as in its universal applicability, lies the strength of Laue's theorem.  

We end this section by a more modern application of Laue's theorem.
Let us consider the question of whether an electromagnetically bound 
system, like a large molecule, obeys a simple version of Einstein's 
equivalence principle, according to which the system's (centre of mass)
acceleration in a given static and weak gravitational field is independent 
of internal energies. As we restrict to weak and static fields, we write
for the components of the space-time metric 
\begin{equation}
\label{eq:WeakART-1}
g_{\alpha\beta}=\eta_{\alpha\beta}+h_{\alpha\beta}\,,
\end{equation}
where $\{\eta_{\alpha\beta}\}=\mathrm{diag}(1,-1,-1,-1)$ represents
the flat Minkowski background and $h_{\alpha\beta}$ deviations from it,
which are considered small, which means higher powers than the first 
of $h_{\alpha\beta}$, as well as of its first two derivatives, are neglected.
The linearized Einstein equations for static Newtonian sources (where only 
the $00$ component of the energy-momentum tensor is assumed non-zero) then 
give, as is well known, 
\begin{equation}
\label{eq:Weak-ART-2}
h_{\alpha\beta}(\vec x)=\delta_{\alpha\beta}\
\frac{2\Phi(\vec x)}{c^2}\,.
\end{equation}
Here $\Phi$ is the standard Newtonian potential and  
$\delta_{\alpha\beta}$ the Kronecker delta (not the Minkowski metric!).
The coupling of a the weak gravitational field $h_{\alpha\beta}$
to the energy momentum tensor of our system (molecule) is given by 
the interaction Lagrangian
\begin{equation}
\label{eq:Weak-ART-3}
L_{\mathrm{int}}=\frac{1}{2}\int_\Sigma d^3x\ T^{\alpha\beta}\,h_{\alpha\beta}=\frac{1}{c^2}\int_\Sigma d^3x\ \Phi
\bigl(T^{00}+T^{11}+T^{22}+T^{33}\bigr)\,.
\end{equation}
If, in addition, we assume that the external gravitational potential
$\Phi$ is approximately constant over the spatial support of 
$T^{\alpha\beta}$, and that the  spatial support at time $t=0$ is 
centred about the spatial position $\vec z$ (e.g. some Newtonian centre 
of mass), we get 
\begin{equation}
\label{eq:Weak-ART-4}
L_{\mathrm{int}}(\vec z)=\Phi(\vec z)\
\int_\Sigma d^3x\ \bigl(T^{00}+T^{11}+T^{22}+T^{33}\bigr)/c^2\,.
\end{equation}
It is clear now that if $\EMT$ satisfies the hypotheses 
of Laues theorem only the first integral survives and 
$L_{\mathrm{int}}(\vec z)$ equals the \emph{inertial mass} of 
the system times the Newtonian potential. This implies the 
\emph{weak equivalence principle} since it explicitly 
demonstrates the equality of the inertial and passive 
gravitational mass (i.e. that mass to which the interaction
with an external gravitational field is proportional to).
Here we note in passing that the combination 
$T^{00}+T^{11}+T^{22}+T^{33}$ equals 
$2(T^{\alpha\beta}-\frac{1}{2}\eta^{\alpha\beta}\eta_{\mu\nu}T^{\mu\nu})n_\alpha n_\beta$ and corresponds to the integrand of 
the so-called Tolman mass \cite{Tolman:1930a}. We refer to 
section\,9 of the arXiv version of
\cite{Giulini:SpringerHandbookSpacetime} for a discussion 
of how the Tolman mass relates to other mass concepts, like 
the Komar mass.

Let us for the moment forget about Laue's theorem and proceed
further in analyzing \eqref{eq:Weak-ART-4}. Our material system 
that is subject to the weak gravitational field consists of 
charged particles interacting via their Coulomb interactions
(we shall neglect radiation). Hence $\EMT=\EMT_1+\EMT_2$, 
with subscript~1 referring to the particles and subscript~2 
to the electromagnetic field. For a particle at rest the 
space integral of $T_1^{00}$ equals $m_0c^2$ and all
other components of vanish. From \eqref{eq:SRT-EMT-Boost-Trans}
we then infer that a single slowly moving particle (in $x^1$ 
direction) has to leading order in $v/c$ a $T_1^{00}$ 
space integral of $m_0c^2+E_{\mathrm{kin}}$, a $T_1^{11}$ space 
integral of $2E_{\mathrm{kin}}$, and all other diagonal components 
vanishing (to leading order). The sum over all particles 
then gives, to leading order in $v/c$,
\begin{equation}
\label{eq:Weak-ART-5}
\int_\Sigma d^3x\,
T_1^{00}=M_0c^2+E_{\mathrm{kin}}\,,\qquad
\int_\Sigma d^3x\,
\delta_{ab}T_1^{ab}=2\,E_{\mathrm{kin}}\,,
\end{equation}
where $M_0$ is the sum of all rest-masses over all particles 
and $E_{\mathrm{kin}}$ is the sum of all kinetic energies over all 
particles. 

On the other hand, the electromagnetic field has a trace-free
energy momentum tensor and hence 
\begin{equation}
\label{eq:Weak-ART-6}
\int_\Sigma d^3x\,
T_2^{00}=U\,,\qquad
\int_\Sigma d^3x\,
\delta_{ab}T_2^{ab}=U\,,
\end{equation}
where $U$ is the total energy stored in the electromagnetic 
field, which we here identify with the total Coulombian 
binding energy, after having subtracted the diverging 
Coulombian self-energy of each particle (compare 
footnote~\ref{footnote:ChargeDist} below).

Hence we get as passive gravitational mass for our molecular 
matter $M_0+(3\,E_{\mathrm{kin}}+2U)/c^2$. But we would have expected
$M_0+(E_{\mathrm{kin}}+U)/c^2$ and now wonder what the excess of 
$(2E_{\mathrm{kin}}+U)/c^2$ might mean. Well, from our derivation we 
know that this combination corresponds to the space integral 
of $\delta_{ab}(T_1^{ab}+T_2^{ab})$, which has to vanish if 
$\EMT=\EMT_1+\EMT_2$ satisfies the hypotheses of Laue's 
theorem (as we assume). But suppose we had never heard of 
Laue's theorem, what could have rescued us from concluding that 
there is something seriously wrong with the equivalence 
principle? Precisely this situation came up in an investigation 
of Carlip's \cite{Carlip:1998}, who concluded that it is the 
virial theorem that rescues us. Recall that in the low-velocity (``non-relativistic'') approximation with homogeneous potentials 
of degree $-1$, that is, any combination of attractive and
repulsive $1/r$ potentials, the virial theorem just implies that 
$2E_{\mathrm{kin}}+U=0$ for the time-averages of kinetic and potential 
energies.%
\footnote{\label{footnote:ChargeDist}%
It is important here to note the following point:
In electrostatics, where (in SI-units) 
$\varepsilon_0\Delta\phi=-\rho$, $\phi$ being the scalar
potential satisfying $\vec E=-\vec\nabla\phi$ and $\rho$ 
being the charge distribution, the electric interaction energy 
for any charge distribution $\rho$ of compact support 
can be written in the following equivalent forms:
\begin{equation}
\label{eq:SelfEnergyRegular}
E_{\mathrm{int}}
=\frac{\varepsilon_0}{2}\int d^3x \Vert\vec E(\vec x)\Vert^2
=\frac{1}{8\pi\varepsilon_0}
\iint d^3 x\,d^3 x'
\frac{\rho(\vec x)\rho(\vec x')}{\Vert\vec x-\vec x'\Vert}\,.
\end{equation}
The first expression shows that  $E_{\mathrm{int}}$ is always positive and cannot possibly 
be identified with $U$ in the virial theorem, where due to 
$2E_{\mathrm{kin}}+U=0$ and $2E_{\mathrm{kin}}\geq 0$ we clearly must 
have $U\leq 0$. On the other hand, if we tried to apply 
\eqref{eq:SelfEnergyRegular} to a number $n$ of point-particles, where  
$\rho(\vec x)=\sum_{a=1}^n q_a\delta^{(3)}(\vec x-\vec x_a)$, 
$E_{\mathrm{int}}$ would clearly diverge due to the self-interaction 
of each individual particle. The second expression in 
\eqref{eq:SelfEnergyRegular} then shows that after 
subtraction of each diverging part the remaining finite 
part is 
\begin{equation}
\label{eq:SelfEnergyParticlesRegularised}
E_{\mathrm{int}}^{\mathrm{(finite)}}
=\frac{1}{4\pi\varepsilon_0}\sum_{a<b}
\frac{q_aq_b}{\Vert\vec x_a-\vec x_b\Vert}\,,
\end{equation}
where the sum runs over all $\frac{1}{2}n(n-1)$ possible 
combinations for $a,b\in\{1,\cdots, n\}$ 
for which $a<b$. That regularised expression 
$E_{\mathrm{int}}^{\mathrm{(finite)}}$ may well assume negative 
values for a mixture of positively and negatively charged 
particles and it is that expression that we may identify 
with $U$ in the virial theorem for such mixed collections 
of charged particles; compare \S\,34 of \cite{Landau.Lifshitz:Vol2}.}
Indeed, that is how we could have concluded without invoking 
Laue's theorem. But I think the foregoing discussion has made 
it clear how elegant and forceful Laue's theorem really is. 
In fact, there is a relation between the derivation of Laue's 
theorem given above and the relativistic virial theorem given 
in  \cite{Landau.Lifshitz:Vol2}, in the sense that the virial 
theorem is a proper implication of Laue's theorem, in that 
it only states the vanishing of the spatial integral of the 
spatial trace of $\EMT$, not the vanishing of the spatial 
integrals of each component of $\EMT$ separately (except $T^{00}$).

\section{Geometric theory}
\label{sec:GeometricTheory}
With this second section, which also contains six
subsections, we turn to the second half of our 
programme as outlined in the introduction. 
Having discussed Laue's theorem and its physical 
relevance in the first section, we now wish 
to know how to formulate it  in a modern
 differential-geometric language. Such a formulation 
should in particular make clear the geometric 
hypotheses underlying it. These hypotheses are not at 
all obvious from the component-based ``derivation'' 
given above, which uses integrals over tensor components 
and identities like \eqref{eq:SRT-LaueThmProof}, 
none of which make any proper geometric sense a priori,
and only receive their restricted meaning in the 
context of SR through the affine and metric structure 
of Minkowski space. In this section we shall give a proper 
geometric statement and proof of a generalisation of 
Laue's theorem that is valid outside SR and reduced to
the classical statement given above if suitably 
specialised. Nothing of what we are now going to 
say depends crucially on the number $n$ of dimensions 
of $M$, so that we shall leave it open. $g$ is then a
 Lorentzian metric on $M$ of signature $(1\,,\,n-1)$, 
i.e. ``mostly minus'' (compare 
Appendix~\ref{sec:Appendix-LorentzianManifolds}). 
Indices are now written in the latin alphabet ranging 
from $0$ to $n-1$.  

\subsection{Alternative representation of 
energy-momentum distributions}
\label{sec:AltRepEMD}
Instead of $\EMT\in\Section(TM\vee TM)$ we define a new object 
by pulling down the first index on $\EMT$ and Hodge-dualising 
on the second:  
\begin{equation}
\label{eq:DefEM-Form}
\EMform:=T_{ab} \theta^a\otimes \star \theta^b
\in\Section\left(T^*M\otimes\Wedge^{n-1} T^*M\right)\,.
\end{equation}
$\EMform$ can be considered as a $T^*M$-valued $(n-1)$-form, or 
equivalently, as an $(n-1)$-form-valued linear map on
$TM$. 

We assume $\EMT$ to be ``complete'' in the sense of Laue's,
by which we mean \eqref{eq:EM-Conservation}, i.e. the 
vanishing of the vector field $\nabla\cdot\EMT:=\nabla_b T^{ab}\,\partial/\partial x^a$ (in $n$ dimensions latin 
indices range from $0$ to $n-1$). Using the exterior 
covariant derivative $D$, the following identity holds: 
\begin{equation}
\label{eq:ExtCovDerEMT-Form}
D\EMform=(\nabla\cdot\EMT)^\flat\otimes\varepsilon\,.
\end{equation} 
Here $\varepsilon$ is the volume $n$-form associated to
$g$ and $\flat$ denotes again the ``index-lowering isomorphism'' defined by $g$
(compare  \eqref{eq:DefMusicalIsomorphisms}).

Let us give a proof of \eqref{eq:ExtCovDerEMT-Form} in 
component language. For that we recall that  the exterior 
covariant derivative $D$ of an $r$-form $F$ with values in 
a vector bundle $V$ over $M$, i.e. of $F\in\Section\bigl(V\otimes\Wedge^rT^*M\bigr)$, 
is obtained from an ordinary torsion-free  covariant 
derivative $\nabla:\Section\bigl(V\otimes\Wedge^rT^*M\bigr)
\rightarrow\Section\bigl(V\otimes T^*M\otimes\Wedge^rT^*M\bigr)$ 
(which in our case is the Levi-Civita connection) by 
total antisymmetrisation in $T^*M\otimes\Wedge^rT^*M$, 
so  as to obtain a map
$D:\Section\bigl(V\otimes\Wedge^rT^*M\bigr)\rightarrow
\Section\bigl(V\otimes \Wedge^{r+1}T^*M\bigr)$. More precisely, 
just writing the form-components (and keeping the 
$V$-valuedness implicit), we have 
$(DF)_{a_0\cdots a_r}=(r+1) \nabla_{[a_0}F_{a_1\cdots a_r]}$. 
Here the  antisymmetrisation bracket $[\cdots]$ is as in \eqref{eq:VollsAntismmIndizes}. The reason for the factor 
$(r+1)$ is the very same as for the corresponding formula 
for the ordinary exterior derivative $d$, which in turn 
is just \eqref{eq:WedgeProdCoeff} applied to $p=1$ and 
$q=r$. Hence, using \eqref{eq:StarMapOnBasis1} and \eqref{eq:VolumeFormSquared}
with $n_-=(n-1)$ (recall our ``mostly-minus'' convention
for the signature in $n$ dimensions), we indeed derive   \eqref{eq:ExtCovDerEMT-Form} in component form through
the following lines:
\begin{equation}
\label{eq:ExtCovDer-Comp}
\begin{split}
 (D\EMform)_{ab_1\cdots b_n}
&=n\,g_{ac}\nabla_{[b_1}T^c_{b_2\cdots b_n]}\\
&=n\,g_{ac}\nabla_{[b_1}T^{cb}\varepsilon_{\vert b\vert b_2\cdots b_n]}\\
&=-(n/n!)\,g_{ac}\varepsilon_{b_1\cdots b_n}\varepsilon^{c_1c_2\cdots c_n}
  \nabla_{c_1}T^{cb}\varepsilon_{bc_2\cdots c_n}\\
&=g_{ac}\varepsilon_{b_1\cdots b_n}\nabla_bT^{cb}\,.
\end{split}
\end{equation}

We conclude from \eqref{eq:ExtCovDerEMT-Form} that 
\eqref{eq:EM-Conservation} is equivalent to 
$\EMform$ having vanishing exterior covariant 
derivative: $D\EMform=0$. However, this does not yet 
define a conservation law. For that we need an ordinary  
$d$-closed $(n-1)$-form. The recipe to get this from 
a $D$-closed $T^*M$-valued $(n-1)$-form is to contract 
the value of the latter with a vector field  
$K\in\Section(TM)$; for that contraction we write 
\begin{equation}
\label{eq:EMform-K}
\EMform_K:=i_K\EMform\,.
\end{equation}
Using  that $D$ and $d$ coincide on ordinary forms, a 
straightforward calculation, very similar indeed to 
that in \eqref{eq:ExtCovDer-Comp}, then shows 
\begin{equation}
\label{eq:ExtCovDer-Current}
d\EMform_K=D(i_K\EMform)=i_K(D\EMform)+\EMT(\nabla K^\flat)\,\varepsilon\,.
\end{equation}
Here the function $\EMT(\nabla K^\flat)$ is the contraction 
of $\EMT\in\Section(TM\vee TM)$ with 
$\nabla K^\flat\in\Section(T^*M\otimes T^*M)$, i.e. 
in components $T^{ab}\nabla_a K_b$. Hence, given symmetry 
\eqref{eq:EM-Symmetry} and $D\EMform=0$ (i.e. \eqref{eq:EM-ConservationSRT}), $\EMform_K$ is closed for general $\EMT$, if and only if 
the symmetric part of $\nabla K^\flat$ vanishes. Now, as
$\nabla$ is the Levi-Civita covariant derivative with respect to 
$g$, we have the well known formula that twice the symmetric 
part of $\nabla K^\flat$ equals $L_Kg$, the Lie derivative of 
$g$ with respect to $K$; in components\footnote{This formula has 
additional terms involving the non-metricity tensor 
$Q:=\nabla g$ and the torsion for general connections.}
\begin{equation}
\label{eq:KillingField}
\nabla_aK_b+\nabla_bK_a=(L_kg)_{ab}\,.
\end{equation}
In other words, given \eqref{eq:EM-Symmetry} and 
\eqref{eq:EM-ConservationSRT}, the $(n-1)$-form $\EMform_K$ 
is closed if and only if $K$ is a Killing vector field, i.e. 
a vector field whose flow is by isometries of $(M,g)$.
It is at this point that both conditions \eqref{eq:EM-Symmetry}
and \eqref{eq:EM-ConservationSRT} enter our discussion 
and further consequences in an essential fashion. 

\subsection{On the notions of ``charges'' and 
``conservation''}
\label{sec:ChargesConservation}
In physics the closedness of the $(n-1)$-form 
$\EMform_K$ is often expressed synonymously by ``conservation'' (i.e. divergencelessness)  
of the ``current'' (i.e. vector field) 
\begin{equation}
\label{eq:EMCurrent}	
 J_K:=(\star\EMform_K)^\sharp
=K_a T^{ab}\partial/\partial x^b\,.
\end{equation}
Then  
\begin{equation}
\label{eq:ClosedEMForm}
d\EMform_K=0\Leftrightarrow \nabla\cdot J_K=0\,.
\end{equation}
Note that $\sharp$ is the inverse of $\flat$ (compare 
\eqref{eq:DefMusicalIsomorphisms}) and also that 
$\star\circ\star$ is the identity on one- and 
$(n-1)$-forms according to \eqref{eq:StarSquared} 
and \eqref{eq:VolumeFormSquared} for mostly-minus 
Lorentz signatures. 

The physical concepts of charges and their conservation  
directly derive from formulae like \eqref{eq:ClosedEMForm}
as follows: Integrate $0=d\EMform_K$ over a $n$-dimensional 
submanifold  $\Omega\subset M$ with piecewise smooth boundary  
$\partial\Omega=\Sigma_1\cup\Sigma_2\cup B$, where $B$ 
is such that $\EMform_K$ restricted to $B$ vanishes; then,
by Stokes' theorem,  
\begin{equation}
\label{eq:EMForm-Conservation}
\int_{\Sigma_1}\EMform_K
=-\int_{\Sigma_2}\EMform_K
=\int_{-\Sigma_2}\EMform_K\,.
\end{equation}
Here $(-\Sigma_2)$ denotes $\Sigma_2$ endowed with opposite 
orientation. If, as in the usual argument,  $\Sigma_1$ 
and $\Sigma_2$ are two spacelike submanifolds whose 
intersection with $\mathrm{supp}(\EMT)$ is compact and  
$B$ is a timelike cylinder connecting the boundaries of
$\Sigma_1$ and $\Sigma_2$ such that 
$B\cap\mathrm{supp}(\EMT)=\emptyset$, then in physics we 
say that the ``charge'' associated with the conserved current 
$J_K$, which is defined to be the flux of $J_K$ through 
the hypersurface, is \emph{conserved} in the sense of 
being independent of the spacelike hypersurface it is 
integrated over. Note that the flux of $J_K$ is obtained 
by integrating $\star(J_K^\flat)$.

We summarise all this in the following definition,
in which we drop the subscript $K$ on $J$ since 
it applies to all $J\in\Section(TM)$, independently 
of whether or not they derive from $\EMT$ by the 
construction above.
\begin{definition}
\label{def:HomologousModulo}
Let $\Omega\subset M$ be a $n$-dimensional 
submanifold with piecewise smooth boundary 
$\partial\Omega=\Sigma_1\cup\Sigma_2\cup B$ 
such that $B\cap\mathrm{supp}(J)=\emptyset$,
then we call $\Sigma_1$ and  $\Sigma_2$
\emph{homologous modulo $\mathrm{supp}(J)$}.
\end{definition}
\begin{definition}
\label{def:Charge}
Let $J\in\Section(TM)$ be any vector field and  
$\Sigma\subset M$ an $(n-1)$-dimensional oriented 
submanifold. Then we call 
\begin{equation}
\label{eq:DefCharge}
Q(\Sigma,J):=\int_\Sigma\star J^\flat
\end{equation}
the \emph{flux of $J$ through $\Sigma$} or, equivalently,
the \emph{charge of $J$ at $\Sigma$}.
\end{definition}
\begin{corollary}
\label{thm:GaussOnFluxes}
If $J\in\Section(TM)$ is divergenceless (is a 
``conserved current'') and $\Sigma_1$ and $\Sigma_2$ 
are homologous modulo $\mathrm{supp}(J)$, that is, 
if there is an $n$-dimensional submanifold 
$\Omega\subset M$ with piecewise smooth boundary 
$\partial\Omega=\Sigma_1\cup\Sigma_2\cup B$ where 
$\mathrm{supp}(J)\cap B=\emptyset$, then 
\begin{equation}
\label{eq:GaussOnFluxes}
Q(\Sigma_1,J)=-Q(\Sigma_2,J)=Q(-\Sigma_2,J)\,,
\end{equation}
where $-\Sigma_2$ stands for $\Sigma_2$ endowed with 
opposite orientation.
\end{corollary}
\begin{proof}
This is an immediate consequence of Stokes' theorem.
\end{proof}

The charges $Q(\Sigma,J)$ are just numbers associated 
to vector fields $J$ and hypersurfaces $\Sigma$. In
physics, these numbers receive their interpretation 
through the interpretation of $J$ and their 
significance through additional special properties 
of $J$, most importantly its vanishing divergence 
and certain constraints on its support.
Corollary\,\ref{thm:GaussOnFluxes} states to what extent
these numbers are independent of $\Sigma$. This leads to 
 the concept of ``conservation'' in the sense that the 
charges measured at different boundaries are the same.
If one of these boundaries is regarded as time
 evolution of the other, then we obtain the ordinary
notion of charge conservation ``in time''.

\subsection{%
Group actions and the habitat of charges}
\label{sec:GroupActionsAndCharges}
To a large degree independent of the issue of conservation is 
the issue of interpretation. As already stated, the  
interpretation of the charges must follow from the interpretation 
of $J$. In our case, the vector fields $J$ we consider derive 
from energy-momentum tensors $\EMT$ and some preferred vector 
field $V$. Hence the interpretation of charge is essentially 
connected with $V$: Which vector field $V$ are we using 
in order to turn $\EMT$ into a ``current'' 
$J=J_V:=i_V\EMT$, the charge of which we consider? 
That is the question we address in this subsection.
We know that if $V=K$ is Killing then $J_K$ is covariant 
divergenceless and the charges are conserved. But for 
the time being we wish to be independent of that assumption 
and hence shall not assume $V$ to be Killing. 

Vector fields $V\in\Section(TM)$ generate 
diffeomorphisms of $M$, which we may physically 
interpret as ``motions'' if acting on geometric 
objects on $M$ representing physical systems. Of ``motions'' 
we usually think as being composable and forming a 
group. It is the group structure of a set of motions 
that we usually employ in order to interpret it elements. 
For example, in Newtonian mechanics we characterise 
an overall motion of a system of point particles as 
being either a ``translation'' or ``rotation'', depending 
on how this particular motion sits inside the group $E_3$ 
of euclidean motions that we think of as having 
been implemented as acting on space (in this case 
by isometries). In $E_3$ translations are invariantly 
characterised as elements of the maximal abelian 
normal subgroup. That is, translations form a unique 
3-dimensional subgroup in the 6-dimensional group 
$E_3$. In contrast,  there is no unique rotation 
subgroup in $E_3$. Rather, $E_3$ contains a 
3-parameter family of different copies of rotation
subgroups corresponding to the different choices 
of origin in space about which one rotates (that 
is kept fixed unter all rotations). These different 
copies are all related by conjugation with translations 
(shifting the origin). Hence we may not speak of 
``the'' but only of ``a'' rotation subgroup in $E_3$,
each one corresponding to its own origin in space. 

Similar remarks apply to the Poincar\'e group in 
special-relativitsic physics. (Compare 
Appendix~\ref{sec:Appendix-PoincareGroup} and 
\cite{Giulini:EMTaMinSRT2015} for a more detailed 
discussion of the semi-direct product structure of 
the Poincar\'e group in terms of a unique subgroup 
of translations and a non-unique complementary 
subgroup of Lorentz transformations.) Here again 
``translations'' in space and time are 
invariantly  characterised as forming the maximal 
abelian normal subgroup of the Poincar\'e group. 
If the translation is timelike we call the corresponding 
charge ``energy'' and ``linear momentum'' if it is spacelike. 
Charges corresponding to any of the Lorentz subgroups, of which 
there is an $n$-parameter family, are then associated to, e.g., 
angular momentum or centre-of-mass motion.  Again picking 
one of the Lorentz subgroups is equivalent to picking a point 
in spacetime, the ``origin'' left fixed by all Lorentz 
transformations.

Let us now turn to the general case and capture the 
situation by supposing that a finite-dimensional 
Lie group $G$ acts on the left on the manifold $M$. 
The words ``left action'' mean that we have a 
homomorphism of groups ($e\in G$ is the group identity):  
\begin{subequations}
\label{eq:IsometricG-Action}
\begin{equation}
\label{eq:IsometricG-Action-a}
\Phi: G\rightarrow\mathrm{Diff}(M)\,, \quad 
g\mapsto\Phi_g\,,\quad
\Phi_e=\id_M\,,\quad
\Phi_g\circ\Phi_h=\Phi_{gh}\,.
\end{equation}
For a ``right action'' the last equation would be 
replaced by $\Phi_g\circ\Phi_h=\Phi_{hg}\,$, but we shall 
stick to the ``left'' convention. 

Later we will consider special actions $\Phi$ of 
$G$ on $M$ which act by isometries with respect 
to a given metric $g$ on $M$. This means that 
\begin{equation}
\label{eq:IsometricG-Action-b}
\Phi^*_hg=g
\end{equation}
\end{subequations}
for each $h\in G$. But for the time being we shall keep 
the discussion independent of that assumptions as long 
as possible. Hence we stress that all statements below 
do not assume the action to be isometric unless 
explicitly stated so.

Any group-homomorphism between Lie-groups induces 
a homomorphism between the corresponding 
Lie-algebas, which is just given by the differential 
of the former at the group identity. The Lie-algebra 
of $\mathrm{Diff}(M)$ is given by the infinite-dimensional 
vector space $\Section(TM)$ of vector fields on $M$ whose 
Lie bracket is \emph{minus} the vector-field commutator.
Hence we have     
\begin{subequations}
\label{eq:LieAntiHomo}
\begin{equation}
\label{eq:LieAntiHomo-a}
V:\Lie(G)\rightarrow\Section(TM)\,,\quad
V_\xi(p):=\frac{d}{ds}\Big\vert_{s=0}\Phi_{\exp(s\xi)}(p)\,,
\end{equation} 
which satisfies 
\begin{alignat}{2}
\label{eq:LieAntiHomo-b}
-\bigl[V_\xi\,,\,V_\zeta\bigr]&\,=\,V_{[\xi,\zeta]}\,,\\
\label{eq:LieAntiHomo-c}
(\Phi_g)_*V_\xi&\,=\,V_{\mathrm{Ad_g}(\xi)}\,.
\end{alignat}
\end{subequations}
We refer to the Appendix of \cite{Giulini:EMTaMinSRT2015} 
or the appendix of the arXiv version of 
\cite{Giulini:SpringerHandbookSpacetime} for a 
detailed proof. Before we continue, we give $V_\xi$ 
a proper name:   
\begin{definition}
\label{def:FundamentalVectorField}
Given an action $\Phi$ of $G$ on $M$, then 
$V_\xi\in\Section(TM)$ is called the 
\emph{fundamental vector field} for that 
action on $M$ corresponding to $\xi\in\Lie(G)$. 
\end{definition}

Note that the bracket $[\cdot,\cdot]$ on the left-hand 
side of \eqref{eq:LieAntiHomo-b} is the vector-field 
commutator, which is \emph{minus} the Lie-bracket in 
the Lie-algebra of $\mathrm{Diff}(M)$, whereas the  
$[\cdot,\cdot]$ on the right-hand side of \eqref{eq:LieAntiHomo-b} 
is the Lie-bracket in $\Lie(G)$. This is the reason why 
there is an additional minus sign on the left-hand side of 
\eqref{eq:LieAntiHomo-b} that makes this relation a Lie
homomorphism.\footnote{Sometimes this minus sign is said 
to express a Lie-\emph{anti}-homomorphism, namely whenever 
$\Lie\bigl(\mathrm{Diff}(M)\bigr)$ is endowed with the 
opposite Lie structure, in which the Lie-bracket is the 
commutator of vector fields.
The minus sign in \eqref{eq:LieAntiHomo-b} would also 
disappear if we considered a right rather than left 
action of $G$ on $M$, though then other minus signs 
would pop up elsewhere. See the appendix in  \cite{Giulini:EMTaMinSRT2015} for 
a detailed discussion and more information.}

Given $\EMT$ we can for each fundamental vector field 
$V_\xi\in\Section(TM)$ and each oriented $(n-1)$-dimensional 
submanifold $\Sigma$ calculate the charge $Q(\Sigma,J)$ 
according to \eqref{eq:DefCharge}. The result is a number 
that depends linearly on $\xi$ since it is obtained by 
composing three linear maps: 
\begin{equation}
\label{eq:ThreeLinearMaps}
\Lie(G)\ni\xi
\mapsto V_\xi
\mapsto\EMform_{V_\xi}
\mapsto\int_\Sigma\EMform_{V_\xi}\in\reals\,.
\end{equation}
In other words, it defines an element in $\Lie^*(G)$,
the vector space dual to $\Lie(G)$. This assignment of 
an element in  $\Lie^*(G)$ obviously depends on $\Sigma$ 
and $\EMT$, but as $\EMT$ depends on the underlying 
physical fields, we shall regard the dependence to be on 
$\Sigma$ and \emph{some} physical fields. Which physical 
fields? Here we shall make the same distinction as in our 
previous discussion in 
Section\,\ref{sec:LocalEnergyMomentumTransformation}.
In the more general geometric setting that we 
consider now, we assume $\EMT$ and hence $\EMform$ to 
be build locally\footnote{Meaning that the value of $\EMform$
at a point $p\in M$ depends on the values of $F$
and $F'$ and at most finitely many derivatives at $p$.}
from a set of fields $F\in\Section(B)$ whose energy-momentum 
distribution is represented by  $\EMT$ (or $\EMform$), and 
possibly a set of complementary background fields 
$F'\in\Section (B')$. Here, $B$ and $B'$ denote some 
natural bundles over $M$. Following standard terminology
(compare Appendix~\ref{sec:Appendix-NaturalBundles}) the word 
``natural'' is added here in order to indicate our 
requirement that any diffeomorphism $\Phi$ of $M$ shall have
a natural action on sections in $B$ and $B'$, which we 
denote by $D$ and $D'$ respectively . We then assume the 
generalisation of  \eqref{eq:SRTBoostedEMT-4}, which now 
reads
\begin{equation}
\label{eq:DiffeoActingOnEMT}
\EMform(D_\Phi F\,,\,D'_\Phi F')=(\Phi^{-1})^*\EMform(F,F')\,.
\end{equation}
It says that $\EMform$ evaluated on the $\Phi$ transformed 
fields $D_\Phi F$ and $D'_\Phi F'$ is the $\Phi$-transform
of $\EMform(F,F')$. Here we already took into account that 
the  $\Phi$-transform of a section in 
$T^*M\otimes\Wedge^{(n-1)}T^*M$ is the pull-back with the 
inverse $\Phi^{-1}$. Condition \eqref{eq:DiffeoActingOnEMT} 
may be seen as a completeness condition, in the sense that 
we did in fact consider all fields $(F,F')$ on which 
$\EMform$ truly depends. This is the generalisation of 
\eqref{eq:SRTBoostedEMT-4} referred to at the end of 
Section\,\ref{sec:LocalEnergyMomentumTransformation}.
As already stressed there, in the context of general 
groups $G$ acting on $M$ it becomes particularly important 
to distinguish between the systems whose states are 
represented by $F$ and whose energy-momentum distribution 
is fully captured by $\EMform$, and the wider system 
whose states are represented by $(F,F')$ on which 
$\EMform$ may actually also depend. 

\begin{definition}
\label{def:MomentumMap}
The map, that according to \eqref{eq:ThreeLinearMaps}
assigns an element in $\Lie^*(G)$ to each triple 
$(\Sigma,F,F')$ is called \emph{momentum map}%
\footnote{The momentum map is a standard tool in 
symplectic geometry and Hamiltonian Mechanics; see, 
e.g., Chapter 4.2 of \cite{Abraham.Marsden:Mechanics}
and references given therein. The difference between 
our and the standard definition is that we start from 
the energy-momentum tensor, which we assume to be given, 
without explicitly using the symplectic geometry of 
phase space.} 
and denoted by $\MM$. We shall put the arguments in square 
brackets and write $\MM[\Sigma,F,F']\in\Lie^*(G)$ since 
it may take another argument from $\Lie(G)$ to turn that 
element into a real number. For $\xi\in\Lie(G)$
we then write $\MM[\Sigma,F,F'](\xi)\in\reals$ and 
note that this number is just the flux of 
$J_{V_\xi}=i_{V_\xi}\EMT$ through $\Sigma$:
\begin{equation}
\label{eq:DefMomentumMap}
\MM[\Sigma,F,F'](\xi)
:=Q\bigl(\Sigma,i_{V_\xi}\EMT(F,F')\bigr)
=\int_\Sigma\EMform_{V_\xi}(F,F')\,.
\end{equation}  
\end{definition}

We wish to know how the value in $\Lie^*(G)$ of the 
momentum map $\MM$ behaves under the action of the 
group $G$. That is answered by 
\begin{theorem}
\label{thm:MomentumMapEquivariance}
Let $(G,\Phi,M)$ be a left action of the finite dimensional 
Lie group $G$ on $M$ and $\EMform\in\Section(T^*M\otimes\Wedge^{(n-1)}T^*M)$ 
a section depending locally on $(F,F')\in\Section(B)\times\Section(B')$
so that \eqref{eq:DiffeoActingOnEMT} is valid for all $\Phi=\Phi_g$, $g\in G$;
then
\begin{equation}
\label{eq:MomentumMapEquivariance}
\MM\bigl[\Phi_g\Sigma,D_{\Phi_g} F,D'_{\Phi_g} F'\bigr]
=\Ad^*_g\bigl(\MM[\Sigma,F,F']\bigr)\,,
\end{equation}
where $\Phi_g\Sigma$ is the Image of $\Sigma$ under $\Phi_g$
and $\Ad^*$ denotes the co-adjoint representation of 
$G$ on $\Lie^*(G)$, i.e. the inverse-transposed of the adjoint
representation. 
\end{theorem}
\begin{proof}
The proof follows from a string of elementary steps:
\begin{equation}
\label{eq:MomentumMapEquivarianceProof}
\begin{split}
\MM\bigl[\Phi_g\Sigma,D_{\Phi_g}F,D'_{\Phi_g}F'\bigr](\xi)
&\stackrel{1}{=}\int_{\Phi_g\Sigma}\EMform_{V_\xi}\bigl(D_{\Phi_g}F,D'_{\Phi_g}F'\bigr)\\
&\stackrel{2}{=}\int_{\Phi_g\Sigma}i_{V_{\xi}}\EMform\bigl(D_{\Phi_g}F,D'_{\Phi_g}F'\bigr)\\
&\stackrel{3}{=}\int_{\Phi_g\Sigma}i_{V_{\xi}}\Bigl(\Phi_{g^{-1}}^*\EMform(F,F')\Bigr)\\
&\stackrel{4}{=}\int_{\Phi_g\Sigma}\Phi^*_{g^{-1}}\Bigl(\EMform_{\Phi_{g^{-1}*}V_\xi}(F,F')\Bigr)\\
&\stackrel{5}{=}\int_\Sigma\EMform_{V_{\Ad_{g^{-1}}}(\xi)}(F,F')\\
&\stackrel{6}{=}\ \MM[\Sigma,F,F'](\Ad_{g^{-1}}(\xi))\\
&\stackrel{7}{=}\ \Ad^*_g\bigl(\MM[\Sigma,F,F']\bigr)(\xi)\,.
\end{split}
\end{equation}
Steps\,1 and 2 just follow the definitions, step\,3 uses 
\eqref{eq:DiffeoActingOnEMT}, step\,4 uses the general 
relation $i_X\circ\Phi^*=\Phi^*\circ i_{\Phi_*X}$ for 
the pull-back and the insertion map applied to forms, 
which is obvious from the fact that the pull-back and 
push-forward are relatively transposed maps, step\,5 
uses \eqref{eq:LieAntiHomo-c} and the general relation 
(``change-of-variable formula'') 
$\int_{\Sigma}\phi^*\alpha=\int_{\phi(\Sigma)}\alpha$ for
the integration of a $k$ form over a $k$-dimensional 
submanifold, step\,6 (like step\,1) invokes the 
definition of $\MM$, and, finally, step\,7 just uses 
the definition of the co-adjoint representation.  
\end{proof}
  
We stress the generality of \eqref{eq:MomentumMapEquivariance},
which in particular does not put any further constraint on 
the action of $G$ on $M$ other than the validity of  
\eqref{eq:DiffeoActingOnEMT}, in particular we did so far not 
require the action to be isometric. Such further conditions
now come into play if we want to strengthen  
\eqref{eq:MomentumMapEquivariance} similarly to the way  
\eqref{eq:SRTBoostedEMT-3} strengthens \eqref{eq:SRTBoostedEMT-4} and \eqref{eq:SRT-LauesThmClassic1}
strengthens \eqref{eq:SRT-LauesThmClassic1-Fake}.
That is, we wish to state \eqref{eq:MomentumMapEquivariance} 
with $G$ just acting on the dynamical fields $F$ and not on 
$\Sigma$ and the background fields $F'$. Hence we wish to have 
\begin{equation}
\label{eq:MomentumMapEquivarianceRestricted}
\MM\bigl[\Sigma,D_{\Phi_g} F,F'\bigr]
=\Ad^*_g\bigl(\MM[\Sigma,F,F']\bigr)\,.
\end{equation}
\emph{This} is the generalisation and, in fact, proper phrasing 
of the requirement that that global ``momenta'' calculated 
from a local distribution of dynamical fields $F$ transform 
as ``vectors'' under the action of a group $G$ on $M$
(a statement to be qualified below). 

So we ask: under what conditions is it true that 
\eqref{eq:MomentumMapEquivariance} implies \eqref{eq:MomentumMapEquivarianceRestricted}? 
We start with the dependence on $F'$ and observe
that $D'_{\Phi_g}F'$ may be replaced by $F'$ if $D'_{\Phi_g}$ 
acts trivially on the background fields $F'$, i.e., 
if $G$ is a symmetry group for all configurations $F'$. 
This is trivially always the case if the set of background 
fields $F'$ is empty, i.e. if the energy-momentum tensor
exclusively only depends on the fields the energy momentum 
distribution of which it represents. This will be hard 
to achieve in general. More common is the case in which it  
depends at least on the background metric $g$, like , e.g., 
in the traditional (metric) formulation of Maxwell's 
electrodynamics. Without external sources no other $F'$ exist, 
so $F'=g$  and $D'_{\Phi_h}F'=(\Phi^{-1}_h)^*g$, which equals 
$g$ for all $h\in G$ if and only if $G$ acts on $M$ 
isometrically with respect to $g$. With external sources the 
background fields $F'$ also include charge densities and 
currents, which then, too, must remain invariant under the action 
of $G$. 

Next we turn to the dependence on $\Sigma$. We note that 
the integral \eqref{eq:DefMomentumMap} would be the same 
for two hypersurfaces $\Sigma_1$ and $\Sigma_2$ if they 
were homologous modulo $\mathrm{supp}(\EMT)$ and if 
$\EMform_{V_\xi}$ were closed for all $\xi\in\Lie(G)$. 
In view of \eqref{eq:ExtCovDer-Current} a necessary and 
sufficient condition for the latter to hold is that 
$V_\xi$ is Killing, i.e. that  $G$ acts as isometries 
on $(M,g)$. From that we infer

\begin{corollary}
\label{thm:EquivarianceGlobalCharge}
The momentum map $\MM$ satisfies the restricted condition 
of equivariance in the form 
\eqref{eq:MomentumMapEquivarianceRestricted} if $G$ acts on 
$(M,g)$ by isometries, if $G$ is a symmetry of all other 
background fields $F'$ (other than the metric),  and if 
$\Phi_h\Sigma$ is homologous to $\Sigma$ modulo 
$\mathrm{supp}(\EMT)$ for all $h\in G$.
\end{corollary}
Note that the hypotheses of this corollary turn the 
values of the momentum map and hence the charges 
(contractions of $\MM$ with elements of $\Lie(G)$) 
into ``conserved'' charges, in the sense that 
the charges at all $\Phi_h\Sigma$ are the same. If 
the group $G$ of motions contain time translations,
the general notion of ``conserved'' used here reduces 
to the usual ``constant in time''.

At the end of this section we wish to stress that the 
hypotheses on mentioned in 
corollary\,\ref{thm:EquivarianceGlobalCharge} are 
sufficient conditions, which need not be necessary in
special situations. We made some effort to first derive 
the general condition of equivariance 
\eqref{eq:MomentumMapEquivariance} under comparatively 
weak hypotheses, and then showed how the standard 
conditions (like that $G$ acts isometrically) lead to 
the restricted  condition of equivariance 
\eqref{eq:MomentumMapEquivarianceRestricted}, known as 
$\Ad^*$-equivariance.

\subsection{Poincar\'e covariant momenta}
\label{sec:PoincareCovMomenta}
The equivariance condition 
\eqref{eq:MomentumMapEquivarianceRestricted}
gives the precise expression to the statement usually 
employed in physics,  that conserved momenta, i.e.
the values of $\MM$, behave ``covariantly'' under 
a certain group of motions. At the same time we learn 
how to characterise the vector space of which the 
conserved momenta are members of. They are elements in $\Lie^*(G)$ and they transform under the co-adjoint representation of $G$ . 

The physical interpretation of these momenta must 
be obtained from $G$. If, e.g., $G$ is the group of 
space-time translations, which is invariantly contained 
in the Poincar\'e group as the maximal abelian normal 
subgroup, then the corresponding ``momenta'' correspond to 
energy and linear momentum. Together they transform like 
``four vectors'' under Poincar\'e transformations, 
that is, trivially under translations and by the 
defining representation of the Lorentz group. But that 
last statement is deceptive, because our momenta are 
really elements of $\Lie^*(G)$, not the four-dimensional 
real vector space $V$ underlying affine Minkowski space. 
In this subsection we show how this is properly phrased 
for the Poincar\'e group and how it relates to our 
previous discussion in Sections\,\ref{sec:GlobalEM-Trans} 
and \ref{sec:LaueClassical}.

To start the discussion, let $G=\Poin$ be the Poincar\'e group
of spacetime, which is isomorphic to $V\rtimes\Lor$, where 
$\Lor$ is the Lorentz group, so that for $(a,A)$ and 
$(b,B)\in V\rtimes\Lor$, we have 
\begin{subequations}
\label{eq:Pin-Laws}
\begin{equation}
\label{eq:Pin-MultLaw}
(a,A)(b,B)=(a+Ab\,,\,AB)
\end{equation} 
and 
\begin{equation}
\label{eq:Pin-Inverse}
(a,A)^{-1}=\bigl(-A^{-1}a\,,\,A^{-1}\bigr)\,,
\end{equation} 
\end{subequations}
where juxtapositions like $Ab$ denote the action of $A\in \Lor$ 
on $b\in V$ under the defining representation. We refer to 
Appendix~\ref{sec:Appendix-MinkowskiSpace} for a discussion of the 
(non natural!) isomorphism $\Poin\cong V\rtimes\Lor$.
From \eqref{eq:Pin-Laws} we can easily calculate the 
adjoint and co-adjoint representation on $\Lie(\Poin)$ and 
$\Lie^*(\Poin)$ respectively. These take a convenient form 
if we identify both of these linear spaces with 
$V\oplus (V\wedge V)$, which also allows for an easy comparison 
between the adjoint and co-adjoint representations (which are 
now realised on the same vector space). Recall that usually one 
identifies $\Poin\cong V\rtimes\Lor$ and $\Lor$ with isometries 
of $(V,\eta)$. This induces the usual identification 
$\Lie(\Poin)\cong V\rtimes\Lie(\Lor)$ where $\Lie(\Lor)\cong\{X\in\End(V):\eta(Xv,w)=-\eta(v,Xw)\,\
\forall v,w\in V\}$. To that we add the identification of 
the latter (the $\eta$-antisymmetric endomorphisms) with 
$V\wedge V$, in that we agree that $x\wedge y\in V\wedge V$ 
corresponds to the following element of $\End(V)$: 
$x\wedge y(v):=x\,\eta(y,v)-y\,\eta(x,v)$, which one easily checks
is $\eta$-antisymmetric. The extension of that correspondence 
to all elements in $V\wedge V$ (not just pure exterior 
products) then follows from linearity. Hence the Lie brackets 
between to elements $(x,X)$ and $(y,Y)$ in $V\oplus (V\wedge V)\cong\Lie(\Poin)$ is 
\begin{equation}
\label{eq:PoinLieBracket}
\bigl[(x,X),(y,Y)\bigr]=
\bigl(X\cdot y-Y\cdot x\,,\,[X,Y]\bigr)\,,
\end{equation}  
where the action of $X,Y\in V\wedge V$ on $V$ (denoted by 
a dot in \eqref{eq:PoinLieBracket}) is as explained above,
and $[X,Y]:=X\circ Y-Y\circ X$ is just the commutator 
of these maps. Finally, we identify $\Lie^*(G)$ with 
$V\oplus (V\wedge V)$ by identifying the dual space of  
$V\oplus (V\wedge V)$ with itself through the following 
inner product (non-degenerate symmetric bilinear form) 
on $V\oplus (V\wedge V)$:
\begin{equation}
\label{eq:PoinLieInnerProduct}
\big\langle(x,X),(y,Y)\big\rangle
=\eta(x,y)+\tfrac{1}{2}\eta\otimes\eta(X,Y)\,.
\end{equation}  
Here the second term just corresponds to what we 
called the renormalised inner product on forms 
in equation \eqref{eq:RenormInnerProduct} of the 
Appendix. Hence if $X=u\wedge v=u\otimes v-v\otimes u$ 
and $Y=w\wedge z=w\otimes z-z\otimes w$ we have $\tfrac{1}{2}\eta\otimes\eta(X,Y)
=\eta(u,v)\eta(v,z)-\eta(u,z)\eta(v,w)$.

The adjoint and co-adjoint representation on 
$V\oplus (V\wedge V)$ are now easily obtained as follows.
Let $s\mapsto \bigl(b(s),B(a)\bigr)$ be a curve in 
$\Poin$ through the identity at $s=0$. 
We set with $d/ds\vert_{s=0}\bigl(b(s),B(s)\bigr)=(P,M)\in V\oplus(V\wedge V)$
and have 
\begin{equation}
\label{eq:Rep-Ad}
\begin{split}
\Ad_{(a,A)}(P,M):=&\frac{d}{ds}\bigg\vert_{s=0}
\bigl(a,A\bigr)\bigl(b(s),B(s)\bigr)\bigl(a,A\bigr)^{-1}\\
=&\Bigl(AP-[(A\otimes A)M]a\,,\,(A\otimes A)M\Bigr)\,.
\end{split}
\end{equation}
The co-adjoint representation is the inverse-transposed of that,
transposition now being understood in terms of the inner product 
\eqref{eq:PoinLieInnerProduct}:
\begin{equation}
\label{eq:Rep-AdTrans-a}
\Bigl\langle(P',M')\,,\,\Ad_{(a,A)}(P,M)\Big\rangle
=\Bigl\langle\Ad^\top_{(a,A)}(P',M')\,,\,(P,M)\Big\rangle\,.
\end{equation}
Writing this out using \eqref{eq:PoinLieInnerProduct}
we read off (dropping the primes)
\begin{equation}
\label{eq:Rep-AdTrans-b}
\Ad^\top_{(a,A)}(P,M)=\Bigl(
A^{-1}P\,,\,\bigl(A^{-1}\otimes A^{-1}\bigr)M
+\bigl(A^{-1}a\bigr)\wedge\bigl(A^{-1}P\bigr)\Bigr)\,,
\end{equation}
and therefore  
\begin{equation}
\label{eq:Rep-CoAd}
\begin{split}
\Ad^*_{(a,A)}(P,M):=&\bigl(\Ad_{(-A^{-1}a,A^{-1})}\bigr)^\top (P,M)\\
=&\Bigl(AP\,,\, (A\otimes A)M-a\wedge AP \Bigr)\,.
\end{split}
\end{equation}
The first entry on the right-hand side of \eqref{eq:Rep-CoAd}
shows what previously was meant when we said that $P=P^\alpha e_\alpha$
transformed as a ``four vector''. But note from the second slot that 
$M$ does not just transform as an ``antisymmetric 2nd-rank tensor'',
there is an additional piece $a\wedge AP$ expressing the base-point
dependence of angular momentum and centre-of-mass motion, which 
group-theoretically exists as a result of the fact that 
$\Lor\subset\Poin$ is not a normal subgroup and hence 
$\Lie(\Lor)\subset\Lie(\Poin)$ not an ideal.  

\newpage
\subsection{Laue's theorem: geometric formulation}
\label{sec:LaueGeometric}
We are now in a position to phrase Laue's theorem in 
geometric language and show what its underlying 
assumptions are. We start with the remark that this 
theorem is, in fact, a statement about conserved 
currents which is then specialised to the case that 
the current is constructed from an energy-momentum 
tensor through contraction with a Killing field. So 
let us start  by assuming that we have a conserved 
current $J\in\Section(TM)$, i.e. in components 
$J=J^\alpha\,\partial/\partial x^\alpha$ with 
$\nabla_\alpha J^\alpha=0$. Its corresponding 
$(n-1)$-form, 
$\star J^\flat\in\Section\bigl(\Wedge^{(n-1)}T^*M\bigr)$
 will be called $\CurrentForm$. Hence $\CurrentForm$ is 
closed (compare 
Appendix~\ref{sec:Appendix-BundlesOverManifolds}),  
\begin{equation}
\label{eq:CurrentFormClosed}
d\CurrentForm=0\,. 
\end{equation}
We assume $\CurrentForm$ to admit a symmetry, that is, we 
assume there exists a vector field $U\in\Section(TM)$ 
such that the Lie derivative of $\CurrentForm$ with 
respect to $U$ vanishes: 
\begin{equation}
\label{eq:CurrentFormSymmetry}
L_U\CurrentForm=0\,. 
\end{equation}
This is equivalent to saying that the motions generated 
by $U$ (its flow) leave $\CurrentForm$ invariant. 

Next we recall Cartan's formula, according to which the
Lie derivative of forms is given by the symmetrised  
combination of the exterior derivative $d$ and the map 
$i_U$ which inserts $U$ in the first tensor slot:
\begin{equation}
\label{eq:CartanEquation}
L_U=d\circ i_U+i_U\circ d\,. 
\end{equation}
Hence \eqref{eq:CurrentFormClosed} and 
\eqref{eq:CurrentFormSymmetry} imply that 
$\CurrentForm_U:=i_U\CurrentForm\in\Section\bigl(\Wedge^{(n-2)}T^*M\bigr)$ 
is closed 
\begin{equation}
\label{eq:CurrentFormContractedClosed}
d\CurrentForm_U=0\,. 
\end{equation}
Clearly that implies by Stokes' theorem that if 
$\sigma_1,\sigma_2\subset M$ are $(n-2)$-dimensional 
submanifolds which are homologous modulo 
$\mathrm{supp}(J)$ then the integrals of 
$\CurrentForm_U$ over $\sigma_1$ equals that of 
$-\sigma_2$. But is not quite yet the result we are 
after. Suppose $\varphi\in C^\infty(M)$ is a smooth 
function, then 
$d\varphi\wedge\CurrentForm_U\in\Section\bigl(\Wedge^{(n-1)}T^*M\bigr)$
is not only closed but exact (compare Appendix~\ref{sec:Appendix-BundlesOverManifolds}): 
\begin{equation}
\label{eq:DefCurrent-UPhi}
\CurrentForm_{(\varphi,U)}
:=d\varphi\wedge\CurrentForm_U
=d(\varphi\CurrentForm_U)\,. 
\end{equation}
Now, this leads us to the sought-for generalised 
Laue-type theorem for conserved currents:
\begin{theorem}[Laue's theorem, geometric version for currents]\
\label{thm:LauesTheoremCurrentVersion}
Let $(M,g)$ be a $n$-dimensional Semi-Riemannian 
manifold, $J\in\Section(TM)$ a divergenceless (with respect 
to the Levi-Civity connection) vector field, and 
$\CurrentForm\in\Section\bigl(\Wedge^{(n-1)}T^*M\bigr)$ 
its corresponding closed $(n-1)$-form. Let Further 
$U\in\Section(TM)$ be a vector field generating 
symmetries of $\CurrentForm$, i.e. $L_U\CurrentForm=0$. 
Then, for any $\varphi\in C^{\infty}(M)$ and any 
$(n-1)$-dimensional submanifold $\Sigma\subset M$ 
such that $\partial\Sigma\cap\mathrm{supp}\bigl(\CurrentForm_{(\varphi,U)}\bigr)=\emptyset$ 
or sufficiently rapid fall-off of $\varphi\CurrentForm$
at each end of $\Sigma$, we have
\begin{subequations}
\label{eq:LauesTheoremGeomVersion}
\begin{alignat}{1}
\label{eq:LauesTheoremGeomVersion-a}
0&\,=\,\int_\Sigma\CurrentForm_{(\varphi,U)}\\
\label{eq:LauesTheoremGeomVersion-b}
 &\,=\,\int_{\Sigma}\bigl(U(\varphi)\star J^\flat-J(\varphi)\star U^\flat\bigr)\\
\label{eq:LauesTheoremGeomVersion-c}
 &\,=\,\int_{\Sigma}\bigl(U(\varphi)\,g(J,n)-J(\varphi)\,g(U,n)\bigr)\,d\mu_\Sigma
\qquad (\text{$\Sigma$ non lightlike})\,.
\end{alignat}
\end{subequations}
Here $\CurrentForm_U:=i_U\CurrentForm$ and $\CurrentForm_{(\varphi,U)}$ 
as in \eqref{eq:DefCurrent-UPhi}. Moreover, $n$ denotes the non-null
normal to $\Sigma$ and the volume form 
$d\mu_\Sigma\in\Section\bigl(\Wedge^{(n-1)}T^*\Sigma\bigr)$ on $\Sigma$ is defined 
by $d\mu_\Sigma=\pm i_n\varepsilon$ with positive/negative sign for 
timelike/spacelike $n$. 
\end{theorem}
\begin{proof}
The first equality \eqref{eq:LauesTheoremGeomVersion-a}
follows immediately from the exactness of 
$\CurrentForm_{(\varphi,U)}$, Stokes' theorem and the hypothesis 
that $\CurrentForm_{(\varphi,U)}$ vanishes on the boundary of 
$\Sigma$ or falls-off sufficiently fast at each end. The second expression \eqref{eq:LauesTheoremGeomVersion-b} 
follows from 
$i_U(d\varphi\wedge\CurrentForm)
=U(\varphi)\CurrentForm-\CurrentForm_{(\varphi,U)}$ 
and the definition of $\star$, given in \eqref{eq:DefHodgeDuality-b} 
of the Appendix, applied to the left-hand side: 
$i_U(d\varphi\wedge\CurrentForm)
=i_U(d\varphi\wedge\star J^\flat)
=(i_U\varepsilon)\,\bigl\langle d\varphi,J^\flat\big\rangle_{\mathrm{norm}}=
J(\varphi)\star U^\flat$. Finally, the third expression 
\eqref{eq:LauesTheoremGeomVersion-c} is valid only in the case the 
normal $n$ to $\Sigma$ is nowhere lightlike, for then 
\eqref{eq:HodgeOfVectors} gives $\star J^\flat=i_J\varepsilon$ and hence  
$\star J^\flat\big\vert_{T\Sigma}
=\bigl(g(J,n)/g(n,n)\bigr)\,i_n\varepsilon\big\vert_{T\Sigma}
=g(J,n)\,d\mu_\Sigma$ 
with $d\mu_\Sigma:=g(n,n)\,i_n\varepsilon\big\vert_{T\Sigma}$. Note that $g(n,n)=\pm 1$, depending on 
whether $n$ is timelike/spacelike ($\Sigma$ spacelike/timelike). 
\end{proof}

\subsection{Recovery of the classical statement of 
Laue's theorem }
\label{sec:RecoveryOfOriginalLaue}
The theorem above is very general, in that for 
given $J$ and $U$ it holds for all $\Sigma$ and 
$\varphi$. In order to see that the classical version
of Laue's theorem emerging as a special case, we now
explicitly list all the particular choices leading to 
it:
\begin{subequations}
\label{eq:LauesThmGeom}
\begin{enumerate}
\item
$(M,g)$ is 4-dimensional Minkowski space, i.e. there exist global affine coordinates $\{x^0,x^1,x^2,x^3\}$ so that 
\begin{equation}
\label{eq:LauesThmGeom-0}
g=\eta=\eta_{\alpha\beta}dx^\alpha\otimes dx^\beta\,
\quad\text{with}\quad
\eta_{\alpha\beta}=\mathrm{diag}(1,-1,-1,-1)\,.
\end{equation}
\item
The symmetry-generating vector field $U$ for $J$ is
that generating time translations, so that 
\begin{equation}
\label{eq:LauesThmGeom-1}
L_U\CurrentForm=0\,\quad\text{where}\quad
U=\frac{\partial}{\partial x^0}\,.
\end{equation}
Note that since $\frac{\partial}{\partial x^0}$ is
also a symmetry of $g=\eta$ the condition 
$L_U\CurrentForm=0$ is equivalent to $L_UJ=0$
(because $L_U$ commutes with index lowering/raising
and also with the Hodge duality map). 
\item
We take $\Sigma:=\{p\in M:x^0(p)=0\}$ so that 
\begin{equation}
\label{eq:LauesThmGeom-2}
n=\frac{\partial}{\partial x^0}\quad\text{and}\quad
d\mu_\Sigma =dx^1\wedge dx^2\wedge dx^3=:d^3x\,.
\end{equation}
\end{enumerate} 
\end{subequations}

Given (1--3) the integrand in 
\eqref{eq:LauesTheoremGeomVersion-c}
becomes (we write 
$\partial\varphi/\partial x^\alpha=:\varphi_{,\alpha}$;
latin indices range from $1$ to $3$; note that $J_0=J^0$
since $\eta_{00}=1$):
\begin{equation}
\label{eq:LauesThmGeom-3}
U(\varphi)\,g(J,n)-J(\varphi)\,g(U,n)=
\varphi_{,0}J_0-J^\alpha\varphi_{,\alpha}=-J^a\varphi_{,a}\,.
\end{equation}
Hence, in this special case, 
Theorem\,\ref{thm:LauesTheoremCurrentVersion} just 
reduces to a more or less obvious result. Indeed, if 
$J^\alpha_{,0}=0$ then $J^\alpha_{,\alpha}=J^a_{,a}$ 
so that the spatial components define a 3-dimensional 
divergenceless vector field in space which vanishes if 
integrated against any gradient vector field, provided 
that the surface terms one encounters in partial 
integration vanish.

So far we have not restricted to any special choice of
$\varphi$, except that the product $\varphi J$ must 
vanish on $\partial\Sigma$ or fall-off sufficiently 
rapidly at $\Sigma$'s ends. These conditions can of 
course always be met by appropriate choices of 
$\varphi$, given any $J$. But if $J$ already has 
support bounded away from $\partial\Sigma$ and/or 
falls off sufficiently fast at each end of $\Sigma$, 
then we may even multiply it with unbounded functions 
$\varphi$, like any of the spatial coordinate functions 
$x^m$, and obtain
\begin{equation}
\label{eq:LauesThmGeom-4}
\int_\Sigma d^3x \, J^a (x^m)_{,a}
=\int_\Sigma d^3x \, J^m=0
\end{equation}
for all $m\in\{1,2,3\}$. Not that in this 
formula the components of $J$ must refer to 
affine (inertial) coordinates. 

Now, the last step in order to arrive at our previous 
classic formulation of Laue's theorem in 
Section\,\ref{sec:LaueClassical} is to apply this 
to the conserved current 
\begin{subequations}
\label{eq:LauesThmGeom-Cont}
\begin{equation}
\label{eq:LauesThmGeom-5}
J= J_K:=i_{K^\flat}\EMT
 = K_\alpha \EMT^{\alpha\beta}\,
   \frac{\partial}{\partial x^\beta}\,,
\end{equation}
where $K$ is Killing. Since we are in Minkowski space 
we have 10 Killing fields at our disposal, but in order 
to reproduce equation \eqref{eq:SRT-VanishingIntegrals} 
we just need to consider the space-time translations
\begin{equation}
\label{eq:LauesThmGeom-6}
K=\frac{\partial}{\partial x^\mu}\,.
\end{equation}
\end{subequations}
Indeed, inserting \eqref{eq:LauesThmGeom-5}
and \eqref{eq:LauesThmGeom-6} into 
\eqref{eq:LauesThmGeom-4} immediately leads to 
\eqref{eq:SRT-VanishingIntegrals}. 

This derivation of the original result 
\eqref{eq:SRT-VanishingIntegrals} might look like 
cracking a walnut with a sledgehammer. However, 
looking at the 3+2 equations \eqref{eq:LauesThmGeom}
and \eqref{eq:LauesThmGeom-Cont} make it clear that 
this is a very special case indeed and that the general 
result it derived from, i.e. 
\eqref{eq:LauesTheoremGeomVersion}, is far more 
general. Also, the derivation of 
\eqref{eq:SRT-VanishingIntegrals} from 
\eqref{eq:LauesTheoremGeomVersion} makes it clear what 
geometric structures are necessary in order to make 
sense of the original form. The physical 
interpretation and significance of other consequences 
of \eqref{eq:LauesTheoremGeomVersion} remain to be 
developed. One could, e.g., think of stationary 
axisymmetric spacetimes in which $U$ is taken to be 
the timelike and $K$ the rotational Killing field 
(which commute, $[U,K]=0$). 
A stationary $\EMT$, i.e. $L_U\EMT=0$, then implies 
$L_U\CurrentForm=0$ with $\CurrentForm=\star (i_{K^\flat}\EMT)^\flat$,
as required. I do not know of attempts to evaluate 
the consequences of \eqref{eq:LauesTheoremGeomVersion} 
in such situations. At the same time is is likely that 
such consequences are implicit in some of the works on 
the dynamics of matter in stationary axisymmetric spacetimes, 
like Kerr.

\newpage
\section{Conclusions and Outlook}
\label{sec:Conclusion}
The underlying theme of this investigation consists in the 
problem of how local distributions of energy and momentum 
meaningfully combine (``add'')  to global quantities. In 
Section\,\ref{sec:GroupActionsAndCharges} we gave an answer
that required a global group action on the state space 
of the physical system under consideration. The group and its 
action links the local geometric objects with the global 
quantities whose habitat and interpretation are closely 
related to group-theoretic concepts. If the action of the 
group on state space derives from an action on space-time, 
the global quantities receive their interpretation from 
the type of motions generated by the group. This is encoded
in the \emph{momentum map} of Definition\,\ref{def:MomentumMap},
together with its property  
\eqref{eq:MomentumMapEquivarianceRestricted} of 
restricted $\mathrm{Ad^*}$-equivariance under motions of the 
dynamical fields alone. Only in combination with restricted  
$\mathrm{Ad^*}$-equivariance can we properly associate the 
local dynamical fields $F$ with globally conserved quantities 
and identify their mathematical habitat.  
For that it is essential that the representatives of 
physical states, i.e. the local fields $F$, are geometric 
objects in the sense that they are naturally acted upon by 
diffeomorphisms.\footnote{%
A pragmatic definition of a local 
geometric object on a manifold is that for any chart it 
may be fully represented in the chart domain by arrays 
of numbers, with well defined transformation rules under 
changes of charts; see, e.g., \S\,4.13 of Trautmann's 
lecture~\cite{Trautman:1965}. This comprises all sections 
in bundles associated to the bundle of linear frames, 
including tensors and their densities, which all transform 
linearly, and also connections, which change by affine 
transformations. For a more systematic approach to 
define geometric structures, see, e.g., \cite{Salvioli:1972}.}
Without such a group-theoretic link no obvious meaning 
can generally be given to integrals over local energy-momentum 
distributions. 

Even though we only considered matter fields in fixed 
background space-times (which may or may not satisfy 
Einstein's equations), our considerations can be extended 
to GR where the metric field itself is 
dynamical and contributes to the global quantities that are 
usually identified with energy and momentum. In that context 
presumably the first person to seriously wonder about 
the habitat and transformation properties of formally 
constructed global quantities was Felix Klein in 
\cite{Klein.Felix:1918b}, who coined the term 
``free affine vector'' to characterise the four numbers 
that resulted from integrations of components of local 
geometric objects (involving the metric and the connection).
The word ``free''  and ``affine'' were meant to say that 
these four numbers are the component of a vector, but 
that vector is not attached to any point in spacetime. 
The word ``vector'' was meant to indicate that these four 
numbers transform properly under a restricted
class of diffeomorphisms. The latter were defined as those 
preserving the metric structure at large spatial distances,
which was assumed to be Minkowskian. Hence, modulo those 
diffeomorphisms that tend to the identity at spacelike 
infinity and which are considered mere gauge transformations 
(redundancies in the description), the Poincar\'e group 
should result with all its associated conserved charges. 
Hence the idea is to construct a momentum map for the 
Poincar\'e group in GR as well, given that
the field configurations satisfy suitable asymptotic 
conditions. This idea is to a certain extend realised by the 
so-called ADM constructions, as explained in more detail in 
\cite{Fischer.Marsden:1979-1} and very lucidly in 
\cite{Beig.Murchadha:1987} (see also 
\cite{Giulini:SpringerHandbookSpacetime} for a general 
review of the ADM construction, its relation to momentum 
maps, and all the relevant references). Again, and now
for the last time, we stress that such expressions 
loose their physical interpretation if formally 
transcribed to situations without a proper group-theoretic 
context. This, it seems to me, is often not sufficiently 
appreciated.   

\newpage
\appendix
\appendixpage

The following three appendices combine general 
information on our notation and conventions in 
differential geometry, as well as some basic but 
conceptually important structural properties 
of Minkowski space and its automorphism group,
the Poincar\'e group.

\section{Minkowski space and Poincar\'e group}
\label{sec:Appendix-MinkowskiSpace}
Within the realm of Special Relativity, spacetime 
is modelled by Minkowski space. In this appendix 
we wish to briefly recall some of the mathematical 
structures underlying Minkowski space and its 
automorphism group, called the Poincar\'e group.
A proper awareness of these structures is important 
for the discussion in the main text.  

For the sake of generality we shall give our discussion 
for spacetimes of general dimension $n$. This will not 
induce any additional complications. Now, $n$-dimensional 
Minkowski space is defined to be a real affine space of 
dimension $n\geq 2$ whose associated vector space is endowed 
with a non-degenerate symmetric bilinear form of signature 
$(1,n-1)$ [i.e. one positive, $n-1$ negative dimensions], 
also called a Lorentzian inner product. Minkowski space 
is also a real $n$-dimensional Manifold homeomorphic 
to $\reals^n$ with a Lorentzian metric and associated 
Levi-Civita connection which is flat. 

\subsection{Affine spaces}
\label{sec:Appendix-AffineSpaces}
Here we wish to concentrate on the first definition,
which immediately connects to the basic physical idea 
of ``inertial motion'',  i.e. the motion of force-free 
extensionless test particles,  which is represented by 
the straight-lines in affine space.  We recall that an 
$n$-dimensional affine space is a triple $(M,V,+)$,
where $M$ is a set, $V$ is an $n$-dimensional real 
vector space and the familiar symbol $+$ denotes an 
action of $V$ -- considered as abelian group -- on 
the set $M$ which is simply transitive. The word 
``action'' means that there is a map 
$M\times V\rightarrow M$ denoted by $(m,v)\mapsto m+v$ 
such that $m+0=m$ for all $m\in M$ and 
$m+(v+v')=(m+v)+v'$ . Note that in this last equation 
the first $(+)$-sign on the left-hand side denotes 
$V$'s action on $M$ whereas the second denotes 
vector addition. On the right-hand side 
both $(+)$-signs denote $V$'s action on $M$. 
The word ``transitive'' means that for each 
pair $(m,m')\in M\times M$ there exists a $v\in V$
such that $m'=m+v$. The word ``simply'' then further 
requires that this $v\in V$ be unique.
Given that, we may extend our suggestive notation 
and write $v=m'-m$ for that unique element $v\in V$ 
which, when acting on $m$, gives $m'$. Hence we 
may write equations like $(m'-p)+(p-m)=m'-m$ or 
$m+(m'-p)=m'+(m-p)$, valid for all $m',m,p$ in $M$.

To sum up, we may ``add'' vectors to points in affine 
space to get new point in affine space, and also 
``subtract'' two points in affine space to get a 
vector. These two operations behave in a natural 
``associative'' way as exemplified by the equations 
above. But note that points in affine space cannot 
be added. It is true that every point $o\in M$ 
defines a bijection 
\begin{equation}
\label{eq:BijectionM-V}
\phi_o:M\rightarrow V\,,\quad
       m\mapsto \phi_o(m):=m-o\,.
\end{equation}
One might use $\phi_o$ to pull back the linear
 structure from $V$ to $M$, so that ``addition'' in 
$M$ could be defined by  
$m+m':=\phi^{-1}_o\bigl(\phi_o(m)+\phi_o(m')\bigr)=o+(m-o)+(m'-o)$.
 However, the result of ``adding'' $m$ to $m'$ in 
this fashionthen then depends on the choice of $o$. 
In fact, any $p\in M$ could be made the result of 
adding $m$ to $m'$: Just choose $o=m+(m'-p)=m'+(m-p)$.  

\subsection{Affine automorphisms and affine groups}
\label{sec:Appendix-AffineAutGroups}
An affine automorphisms of an affine space is a bijection
$F:M\rightarrow M$ of the following form: There exists a 
point $m\in M$ and a map  $f\in\mathrm{GL}(V)$ (here and 
below $\mathrm{GL}(V)$ denotes the general linear group 
of $V$, i.e. the group of all linear invertible self-maps 
of $V$)), such that $F$ has the form    
\begin{equation}
\label{eq:DefAffineAutom}
F(m+v)=F(m)+f(v)
\end{equation}
It immediately follows that once \eqref{eq:DefAffineAutom}
is true for some $m$ and some $f$ then it also holds for 
all $m\in M$ with the same $f$. Hence $f$ is uniquely 
determined by $F$. Therefore, if $m'+v'=m+v$, we may either 
write $F(m')+f(v')$ or $F(m)+f(v)$ for the image point. 
 
Affine maps form a group under composition which we 
denote by $\Aff(M,V,+)$. Given two affine maps $F$ 
and $F'$ with associated linear maps $f$ and $f'$, the 
composition $F'\circ F$ has associated linear map $f'\circ f$
as is easily checked. Moreover, if $F=\id_M$
then $f=\id_V$. Hence the map
\begin{equation}
\label{eq:ProjectionAff-GL}
\begin{split}
\pi:\Aff(M,V,+)&\rightarrow \GL(V)\\
F&\mapsto \pi(F)=f\,,
\end{split}
\end{equation} 
which assigns to any affine map $F$ its associated linear 
map $f$, is a surjective homomorphism of groups. The kernel 
of $\pi$ is given by an abelian normal subgroup 
$\Trans(M,V,+)\subset Aff(M,V,+)$ which one calls the 
translations of $M$ and which is naturally isomorphic with 
$V$. The quotient $\Aff(M,V,+)/\Trans(M,V,+)$ is isomorphic 
to $\GL(V)$. In fact, $\Aff(M,V,+)$ is isomorphic to the 
semi-direct product $V\rtimes\GL(V)$ in which $\GL(V)$ 
acts as automorphism of the group $V$ by its defining 
representation. This means that  for $(a,A)$ and $(b,B)$ in 
the set $V\times\GL(V)$, their multiplication and inversion 
in $V\rtimes\GL(V)$ is 
\begin{subequations}
\label{eq:DefSemiDirectProduct}
\begin{alignat}{1}
\label{eq:DefSemiDirectProduct-a}
(a,A)(b,B)\,&=\,(a+Ab,A\circ B)\,,\\
\label{eq:DefSemiDirectProduct-b}
(a,A)^{-1}\,&=\,(-A^{-1}a,A^{-1})\,.
\end{alignat}
\end{subequations}

It is important to understand how this isomorphism between 
$\Aff(M,V,+)$ and  $V\rtimes\GL(V)$ comes about, because it 
is not natural due to the fact that it depends on the choice 
of an auxiliary point $o\in M$. Given such a point $o$, a
bijection $\Phi_o$ between the sets $\Aff(M,V,+)$ and 
$V\times\GL(V)$ can be obtained by a combination of the maps \eqref{eq:BijectionM-V} and \eqref{eq:ProjectionAff-GL}
as follows:
\begin{equation}
\label{eq:IsomAff-SDP}
\begin{split}
\Phi_o:\Aff(M,V,+)&\rightarrow V\times\GL(V)\\
F&\mapsto\Phi_o(F)
:=\Bigl(\phi_o\bigl(F(o)\bigr)\,,\,\pi(F)\Bigr)\,.
\end{split}
\end{equation}
Given two elements $F,F'$ in $\Aff(M,V,+)$, we write 
$\Phi_o(F)=(v,f)$ and $\Phi_o(F')=(v',f')$ and have 
$\Phi_o(F'\circ F)=\bigl(v'+f'(v)\,,\,f'\circ f\bigr)$,
as a short computation shows. This means that the 
bijection \eqref{eq:IsomAff-SDP} becomes an isomorphism 
of groups if $V\times \GL(V)$ is endowed with the semi-direct 
product structure \eqref{eq:DefSemiDirectProduct}. 
Note that both, $V$ as well as $\GL(V)$, are subgroups of 
$V\times \GL(V)$. $V$ is the invariant subgroup 
$\mathrm{kernel}(\pi)$ and $\GL(V)$ is the (non-invariant) 
subgroup of those elements fixing the point $o\in M$. 
Hence, in contrast to the translations, the linear part   
$GL(V)$ has no natural place in $\Aff(M,V,+)$.

\subsection{Affine bases and affine charts}
\label{sec:Appendix-AffineBasesCharts}
An \emph{affine basis} $B$ of an affine space consists of a 
tuple $B=(o,b)$, where $o\in M$ and $b=\{e_1,\cdots e_n\}$ 
is a basis of $V$. It defines a bijection of sets 
$\phi_B: M\rightarrow\reals^n$ through $\phi_B(m):=\bigl(x^1(m),\cdots,x^n(m)\bigr)$, 
where $x^a(m):=\theta^a(m-o)$. The inverse map 
$\phi^{-1}_B:\reals^n\rightarrow M$ is given by 
$\phi^{-1}_B(x^1,\cdots, x^n)=o+x^ae_a$. 
The bijection $\phi_B$ becomes a topological homeomorphism 
if the topology given to $M$ is the initial one, i.e. the 
weakest (coarsest) to make $\phi_B$ continuous (the continuity 
of $\phi^{-1}_B$ follows from the invariance-of-domain theorem). 
Moreover, $\phi_B$ defines a global chart (i.e. an atlas with a 
single chart) which can be used to endow  $M$ with the structure 
of a differentiable manifold. We call such a chart an 
\emph{affine chart}. Given two different affine bases, $B=(o,b)$ 
and $B'=(o',b')$, the transition function between the corresponding 
affine charts is   
\begin{subequations}
\label{eq:AffineChartChange}
\begin{equation}
\label{eq:AffineChartChange-a}
\phi_{BB'}:=\phi_B\circ\phi^{-1}_{B'}:
\reals^n\rightarrow\reals^n\,,
\end{equation}
whose $a$-th component is 
\begin{equation}
\label{eq:AffineChartChange-b}
\phi^a_{BB'}(x^1,\cdots,x^n)=A^a_bx^b+a^a\,,
\end{equation}
\end{subequations}
where $a^a:=\theta^a(o'-o)$ and $A^a_b=\theta^a(e'_b)$, i.e. 
$e_b=A^a_be_b$.  

\subsection{Passive versus active transformations}
\label{sec:Appendix-PassiveActive}
In the physics literature these transition function between 
different affine charts are sometimes referred to as 
``passive transformations''. In contrast, an 
``active transformation'' then refers to an affine map of 
$M$ represented in terms of coordinates with respect 
to one and the same chart. So given an affine map $F$ and 
a single affine chart $\phi_B$, we have 
\begin{subequations}
\label{eq:AffineMapCoordinates}
\begin{equation}
\label{eq:AffineMapCoordinates-a}
F_B:=\phi_B\circ F\circ\phi_B^{-1}\,:\,
\reals^n\rightarrow \reals^n\,,
\end{equation}
whose $a$-th component is 
\begin{equation}
\label{eq:AffineMapCoordinates-b}
F^a_B(x^1,\cdots,x^n)=A^a_bx^b+a^a\,,
\end{equation}
\end{subequations}
where $a^a:=\theta^a\bigl(F(o)-o\bigr)$ 
and $A^a_b:=\theta^a\bigl(f(e_b)\bigr)$, i.e.
$f(e_b)=A^a_be_a$.
The analytic similarity of the right-hand sides of 
\eqref{eq:AffineChartChange-b} and 
\eqref{eq:AffineMapCoordinates-b} should not deceive 
us about their different meanings. Whereas 
\eqref{eq:AffineChartChange-b} relates coordinate 
representatives of the same point in different affine 
charts, \eqref{eq:AffineMapCoordinates-b} relates 
coordinate representatives of source and image point 
under an affine map in one and the same affine chart. 
This remark extends to general tensor fields, where 
one can either relate coordinate representatives of 
the same field in different affine charts, or 
coordinate representatives of different fields in 
the same affine chart. Here the ``different fields''
are related by the action of the affine group 
(or subgroup thereof) on the particular bundle in which 
the field is a section; compare 
Appendix~\ref{sec:Appendix-NaturalBundles}.
Note that changing the chart will result in a change 
of the coordinate representatives of \emph{all} fields, 
whereas we may choose to let affine transformation just 
act on any proper subset of fields.  Also, on general 
manifolds, passive transformations will not form a 
group unless all charts have the same domain.

\subsection{Poincar\'e group as metric preserving 
affine group}
\label{sec:Appendix-PoincareGroup}
Finally we mention the definition of the Poincar\'e
group, which we embed into a more general setting. 
Assume $V$ to be endowed with a ``metric'' $g\in V^*\otimes V^*$, 
that is, a non-degenerate symmetric bilinear form of 
any signature. This defines the generalised orthogonal 
subgroup, 
\begin{equation}
\label{eq:GenOrthogGroup}
\mathrm{O}(V,g):=\{A\in\GL(V) : 
g(Av,Aw)=g(v,w)\,\forall v,w\in V\}\,.
\end{equation}
The \emph{inhomogeneous generalised orthogonal group}, 
which we denote by $\mathrm{IO}(M,V,+,g)$, is then 
defined as the preimage in $\Aff(M,V,+)$ of 
$\mathrm{O}(V,g)\subset\GL(V)$ under the homomorphism \eqref{eq:ProjectionAff-GL}:
\begin{equation}
\label{eq:GenInhomogOrthogGroup}
\mathrm{IO}(M,V,+,g):=\pi^{-1}\bigl(\mathrm{O}(V,g)\bigr)\subset\Aff(M,V,+)\,.
\end{equation}
$\mathrm{IO}(M,V,+,g)$ is the automorphism group of the 
structure $(M,V,+,g)$ that comprises the affine structure 
$(M,V,+)$ with the metric structure $g$ on $V$. 
$\mathrm{IO}(M,V,+,g)$ is (not naturally) isomorphic to 
$V\rtimes O(V,g)$, as is immediately seen by restricting 
any of the isomorphisms $\Phi_o$ in \eqref{eq:IsomAff-SDP} to 
$\mathrm{IO}(M,V,+,g)\subset\Aff(M,V,+)$. Now, if $g$ is of 
signature $(1,(n-1))$ [or $((n-1),1)$], in which case we 
write $\eta$ instead of $g$, then $\mathrm{O}(V,g=\eta)$ 
is called the Lorentz group $\Lor$, and $\mathrm{IO}(M,V,+,g=\eta)$ 
is called the Poincar\'e group, $\Poin$, in $n$ dimensions. 
Hence $\Poin:=\pi^{-1}\bigl(O(V,g)\bigr)$.

\section{Spacetime: General conventions and notation}
\label{sec:Appendix-ConventionsNotation}
\subsection{The notion of spacetime}
\label{sec:Appendix-Spacetime}
In physics the notion of ``spacetime'' is usually 
taken to imply a tuple $(M,g)$ in which the first entry,
$M$,  is a  $n\geq 3$-dimensional (usually $n=4$) 
differentiable manifold which is connected, Hausdorff, 
paracompact, and orientable. The second entry, 
$g$, represents a Lorentzian metric with respect to which 
$M$ is time orientable and usually (but not always) 
assumed to be globally hyperbolic. These latter two 
properties refer to the pair $(M,g)$ and are discussed 
in a little more detail below. Let us begin by discussing 
aspects of $M$ alone.    

\subsection{Some obvious bundles over manifolds}
\label{sec:Appendix-BundlesOverManifolds}
As usual, the tangent and cotangent spaces at a point 
$m\in M$ are denoted by $T_mM$ and $T^*_mM$, respectively 
and the tangent and cotangent bundles by $TM:=\bigcup_{m\in M}T_mM$ 
and $T^*M:=\bigcup_{m\in M}T^*_mM$, respectively. The tensor 
product of the $p$-fold tensor product of $T_mM$ with the 
$q$-fold tensor product of $T_m^*M$ is called 
$T_{m\,q}^{\phantom{m}p}M$. It is said to consist of $p$-fold 
contravaraint and $q$-fold covariant tensors at $m\in M$. 
The bundle of such tensors is denoted by $T^p_qM:=\bigcup_{m\in M}T_{m\,q}^{\phantom{m}p}M$. The symbols $\otimes$, $\vee$, 
and $\wedge$ denote the tensor product, the symmetrised, 
and antisymmetrised tensor product, respectively. 
We shall mainly use the antisymmetrised tensor product which is 
explicitly defined in \eqref{eq:DefWedgeProd}. The symmetrised 
product is defined analogously just without the 
sign-function $sign(\sigma)$ in \eqref{eq:DefAlt-Operator}.  
We shall also make use of the inner-product structure that 
the antisymmetric tensor product inherits from that of its 
factors. Since here various conventions concerning 
combinatorial factors as well as signs enter, which will be 
important to keep in mind in order to compare formulae 
from different sources, we have collected all our conventions 
and constructions in a self-contained exposition in 
Appendix~\ref{sec:Appendix-MultilinearAlgebra}. 
This includes in particular 
the Hodge duality map, that we shall denote as usual by 
$\star$. The $p$-fold antisymmetric tensor product of 
$T^*_mM$ is denoted by $\Wedge^pT^*_mM$ and the bundle of 
such forms accordingly by $\Wedge^pT^*M$. A section
$\Lambda\in\Section\bigl(\Wedge^pT^*M\bigr)$ is called
\emph{closed} if its exterior differential vanishes:
 $d\Lambda=0$. It is called \emph{exact} if $\Lambda=d\lambda$
for some  
$\lambda\in\Section\bigl(\Wedge^{p-1}T^*M\bigr)$

\subsection{Natural bundles}
\label{sec:Appendix-NaturalBundles}
If $B$ denotes any bundle over $M$, including all the 
ones mentioned so far, we denote by $\Section(B)$ the 
linear space of all its smooth (say $C^\infty$) sections. 
Dynamical laws in physics are differential equations that 
elements of $\Section(B)$ have to satisfy. In this paper we 
are interested in \emph{symmetries} of such equations, 
that is, operations on $\Section(B)$ under which the 
subset of solutions stays invariant (as set). Of particular 
interest for us are symmetries that derive from symmetries 
of spacetime $(M,g)$, e.g., an action $\Phi$ of a 
finite-dimensional Lie group $G$ on $M$ in the sense of \eqref{eq:IsometricG-Action}. In order to meaningfully ask 
whether this action defines a symmetry of a differential 
equation on  $\Section(B)$ we must first know how the 
action $\Phi$ lifts to an action $\hat\Phi$ on 
$\Section(B)$, i.e. a homomorphism 
$\hat\Phi:G\rightarrow\mathrm{Diff}(B)$ such that for all
$g\in G$ we have $\pi\circ\hat\Phi_g=\Phi_g\circ\pi$, 
where $\pi:B\rightarrow M$ denotes the bundle projection. 
Now, a bundle $B$ is called \emph{natural} if such a lift 
exists naturally for all $\Phi\in\mathrm{Diff}(M)$~\cite{Kolar.etal:NatOpDiffGeom}. 
This is, e.g.,  the case for all bundles discussed above: 
Just take the push-forward of $\Phi$ on vector fields, the 
pull-back of  $\Phi^{-1}$ on form fields, and extend naturally 
to tensor products of those. Generally it is true that $\Phi$ 
naturally lifts from $M$ to the principal bundle $B=FL(M)$ 
of linear frames (by push-forward of each frame) and hence 
to all bundles associated to $FL(M)$. In particular this 
implies that the Lie derivative is defined for all sections 
in natural bundles. On the other hand, there exist physically 
relevant bundles which are not natural in that sense, like, e.g., 
the principal bundle of orthogonal (with respect to some metric 
$g$) frames $FO(M,g)$ and its associated bundles, and also the 
double cover $\overline{FO}(M,g)$, an important associated bundle 
of which is the bundles of spinors. In these cases natural lifts 
exist only for the subgroup of isometries $(M,g)$ within 
$\mathrm{Diff}(M)$. This implies that the \emph{ordinary} Lie 
derivative of spinor fields exists only with respect to vector 
fields generating isometries (Killing fields). For our purpose, 
in particular in connection with the discussion of Section\,\ref{sec:GroupActionsAndCharges}, this is sufficient. 
But we mention that there exist generalisations of the 
ordinary notion of a Lie derivative that do apply to spinor 
fields such that derivatives with respect to \emph{all} vector 
fields make sense. Compare, e.g., the construction in \cite{Bourguignon.Gauduchon:1992} and the more systematic
treatment in \cite{Fatibene.Francaviglia:DiffGeom} within the 
framework of gauge theories.

\subsection{Lorentzian manifolds}
\label{sec:Appendix-LorentzianManifolds}
As already anticipated above, we shall endow $M$ with a 
Lorentzian metric $g\in\Section(T^*M\vee T^*M)$. That  
means that $g(m)\in T_m^*M\otimes T_m^*M$ is a non-degenerate 
symmetric bilinear form  of signature $(1,-1,\cdots,-1)$
at each piont $m$. There is a slight conventional ambiguity 
here in that sometimes the opposite signature  $(-1,1,\cdots,1)$
is taken, which for obvious reasons is referred to as 
the ``mostly plus'' convention. In this paper we stick to the 
``mostly minus'' convention. We assume $(M,g)$ to be globally 
hyperbolic, which implies that $M$ is diffeomorphic to 
$\reals\times\Sigma$, where $\Sigma$ is a connected 
$(n-1)$-dimensional Riemannian manifold with Riemannian 
structure $h=-g\vert_{T\Sigma}$. We also assume $(M,g)$
to be time-orientable, which means that there exists 
a vector field $X\in\Section(TM)$ which is 
timelike (i.e. $g(X,X)>0$; hence, in particular, nowhere
vanishing). At each $m\in M$ the vector $X_m\in T_mM$
selects one of the two components of the so-called 
``chronological cone'' 
$\mathcal{C}_m:=\{X\in T_mM:g(X,X)>0\}\subset T_mM$, 
which one may call the \emph{future} component. 

An orientation for $M$ is picked by specifying a volume 
form $\varepsilon\in\Section\left(\Wedge^n T^*M\right)$.
Relative to that choice, a basis $\{e_0(m),\cdots,,e_{n-1}(m)\}$ of 
$T_mM$ is then defined to be positively oriented, if and only if 
$\varepsilon_m\bigl(e_0(m),\cdots,e_{n-1}(m)\bigr)>0$. Let $\{e_\alpha\in\Section(TM): \alpha =0,\cdots, n-1\}$ 
be a basis at each point (such global bases exist, due to 
$M=\reals\times\Sigma$, if $\Sigma$ is parallelizable, which 
is automatic in $n=4$ dimensions, i.e. if $\Sigma$ is 3-dimensional)  
and $\bigl\{\theta^\alpha\in\Section(TM): \alpha =0,\cdots,n-1\bigr\}$ 
be its dual basis, i.e. $\theta^\alpha(e_\beta)=\delta^\alpha_\beta$. 
Then we call this dual pair of bases \emph{adapted} if 
they are orthonormal in the sense of $g(e_\alpha,e_\beta)=\eta_{\alpha,\beta}=\mathrm{diag}(1,-1,\cdots,-1,)$, 
positively oriented in the sense of $\varepsilon(e_0,\cdots, e_{n-1})>0$, 
and time oriented in the sense of $g(X,e_0)>0$. The volume form 
$\varepsilon$ may then be identified with the unique element in 
$\Section\left(\Wedge^n T^*M\right)$ which assigns the value 1 
(unit volume) to any adapted basis: $\varepsilon(e_0,\cdots,e_{n-1})=1$; hence
\begin{equation}
\label{eq:UniqueMeasure}
\varepsilon=
\theta^0\wedge\cdots\wedge\theta^{n-1}\,.
\end{equation}

\section{Definitions and convention in 
multilinear algebra}
\label{sec:Appendix-MultilinearAlgebra}
In this appendix we collect our definitions and conventions
in a self-contained way, so as to ease comparison with the 
main text and also provide some explanatory material.

\subsection{Vector spaces and the exterior 
algebra of its dual space}
\label{sec:Appendix-VectorSpaces}
Let $V$ be a real $n$-dimensional vector space, $V^*$ its 
dual space and $T^pV^*=V^*\otimes\cdots\otimes V^*$ its 
$p$-fold tensor product.\footnote{We follow standard tradition 
to define \emph{forms}, i.e. the antisymmetric tensor 
product on the dual vector space $V^*$ rather than 
on $V$. Clearly, all constructions that are to follow could 
likewise be made in terms if $V$rather than $V^*$.}
$T^pV^*$ carries a representation $\pi_p$ of $S_p$, the 
symmetric group (permutation group) of $p$ objects, 
given by 
\begin{equation}
\label{eq:SymmetricGroupAction}
\pi_P:S_p\rightarrow\End(T^pV^*),\quad
\pi_p(\sigma)\bigl(\alpha_1\otimes\cdots\otimes \alpha_p\bigr):=
\alpha_{\sigma(1)}\otimes\cdots\otimes \alpha_{\sigma(p)}
\end{equation}
and linear extension to sums of tensor products. 
On  $T^pV^*$ we define the linear operator of 
antisymmetrisation by  
\begin{equation}
\label{eq:DefAlt-Operator}
\Alt_p:=\frac{1}{p!}\sum_{\sigma\in S_p}\mathrm{sign(\sigma)\,\pi_p}\,,
\end{equation} 
where $\mathrm{sign}:S_p\rightarrow\{1,-1\}\cong\integers_2$ is
the sign-homomorphism. This linear operator is idempotent (i.e. 
a projection operator) and its image of $T^pV^*$ under $\Alt_p$ 
is the subspace of totally antisymmetric tensor-products. 
We write
\begin{equation}
\label{eq:DefExtProductSpace}
\pi_p\bigl(T^pV^*\bigr)=:\Wedge^pV^*\,.
\end{equation} 
Clearly 
\begin{equation}
\label{eq:DimExtProductSpace}
\dim\left(\Wedge^pV^*\right)=
\begin{cases}
\binom{n}{p} & \mathrm{for}\ p\leq n\,,\\
0&\mathrm{for}\ p>n\,.
\end{cases}
\end{equation} 
We set 
\begin{equation}
\label{eq:DefExtAlgebra}
\Wedge V^*:=\bigoplus_{p=0}^n\Wedge^p V^*\,.
\end{equation}
Let $\alpha\in\Wedge^pV^*$ and $\beta\in\Wedge^qV^*$, then we define 
their antisymmetric tensor product 
\begin{equation}
\label{eq:DefWedgeProd}
\alpha\wedge\beta:=\tfrac{(p+q)!}{p!q!}\,\mathrm{Alt}_{p+q}(\alpha\otimes\beta)\in\Wedge^{p+q}V^*\,. 
\end{equation}
One easily sees that 
\begin{equation}
\label{eq:WedgeProdSymm}
\alpha\wedge\beta=(-1)^{pq}\,\beta\wedge\alpha\,.
\end{equation}
Bilinear extension of $\wedge$ to all of $\Wedge V^*$ 
endows it with the structure of a real $2^n$-dimensional 
associative algebra, the so-called exterior algebra over $V^*$.   
If $\alpha_1,\cdots,\alpha_p$ are in $V^*$, we have  
\begin{equation}
\label{eq:WedgeOfOneForms}
\alpha_1\wedge\cdots\wedge \alpha_p
=\sum_{\sigma\in S_p}\text{sign}(\sigma)\,
\alpha_{\sigma(1)}\otimes\cdots\otimes \alpha_{\sigma(p)}\,,
\end{equation}
as one easily shows from \eqref{eq:DefWedgeProd} 
and \eqref{eq:WedgeProdSymm} using induction. 

If $\{\theta^1,\cdots, \theta^n\}$ is a basis of $V^*$, 
a basis of $\Wedge^p V^*$ is given by the following 
$\binom{n}{p}$ vectors
\begin{equation}
\label{eq:InducedBasisOnExtAlg}
\{\theta^{a_1}\wedge\cdots\wedge \theta^{a_p}
\mid 1\leq a_1<a_2<\cdots<a_p\leq n\}\,. 
\end{equation}
An expansion of $\alpha\in\Wedge^p V^*$ in this basis is 
written as follows  
\begin{equation}
\label{eq:WedgeProdExpansion}
\alpha=:\tfrac{1}{p!}\,
\alpha_{a_1\cdots a_p}\,\theta^{a_1}\wedge\cdots\wedge \theta^{a_p}\,,
\end{equation}
using standard summation convention and where the coefficients 
$\alpha_{a_1\cdots a_p}$ are totally antisymmetric in all indices. 
On the level of coefficients, \eqref{eq:DefWedgeProd} reads 
\begin{equation}
\label{eq:WedgeProdCoeff}
(\alpha\wedge\beta)_{a_1\cdots a_{p+q}}=\tfrac{(p+q)!}{p!q!}
\alpha_{[a_1\cdots a_p}\beta_{a_{p+1}\cdots a_{p+q}]}\,,
\end{equation} 
where square brackets denote total antisymmetrisation in all 
indices enclosed: 
\begin{equation}
\label{eq:VollsAntismmIndizes}
\alpha_{[a_1\cdots a_p]}
:=\frac{1}{p!}\sum_{\sigma\in S_p}\mathrm{sign}(\sigma)\
\alpha_{a_{\sigma(1)}\cdots a_{\sigma(p)}}\,.
\end{equation}

\subsection{Inner products on vector spaces and 
their extension to the dual space and its exterior 
algebra}
\label{sec:Appendix-InnerProducts}
Suppose there is an inner product (non-degenerate symmetric 
bilinear form) $g$ on $V$, i.e., 
$g:V\times V\rightarrow\reals$; $(v,w)\mapsto g(v,w)$. 
We do not need to restrict to any specific signature 
of $g$ which we therefore leave open. If the signature 
is Lorentzian we shall use the letter $\eta$ instead 
of $g$, as sometimes done in the main text. A non 
generate $G$ defines an isomorphism $\gdown: V\rightarrow V^*$ 
though $\gdown(v):=g(v,\cdot)$. The inverse map is called 
$\gup$, i.e., $\gup:=(\gdown)^{-1}: V^*\rightarrow V$. 
Hence the inner product $g$ on $V$ defines in a natural 
fashion an inner product on the dual space $V^*$, which we 
call $g^{-1}:V^*\times V^*\rightarrow\reals$; 
$(\alpha,\beta)\mapsto g^{-1}(\alpha,\beta)$. It is defined 
in the obvious way by first mapping $\alpha$ and $\beta$ into 
$V$ using $\gup$, and then taking the inner product using 
$g$. Hence $g^{-1}:=g\circ(\gup\times\gup)$, that is, 
$g^{-1}(\alpha,\beta)=g\bigl(\gup(\alpha),\gup(\beta)\bigr)$.
From that it immediately follows that $\gup(\beta)=g^{-1}(\cdot,\beta)$,
if we identify $V^{**}\cong V$ by the natural isomorphism 
$i:V\rightarrow V^{**}$; $i(v)(\alpha):=\alpha(v)$.
The maps $\gup$ and $\gdown$ are often called
the ``index-raising'' and ``index lowering'' map, respectively,
or simply the ``musical isomorphisms''. 
The reason for these names is as follows: Let $\{e_1,\cdots, e_n\}$
$\{\theta^1,\cdots,\theta^n\}$ be dual pairs of bases for $V$ 
and $V^*$ respectively, and $g(e_a,e_b)=:g_{ab}$,
$g^{-1}(\theta^a,\theta^b)=:g^{ab}$. Then 
$\gdown(v^be_b)=v_a\theta^a$ and 
$\gup(\alpha_b\theta^b)=\alpha^ae_a$ with $v_a:=v^b g_{ba}$
and $\alpha^a:=g^{ab}\alpha_b$; also 
$g^{ac} g_{bc}=g^{ca} g_{cb}=\delta^a_b$, which 
justifies the notation $g^{-1}$, since for 
symmetric $g$ the matrix $\{g^{ab}\}$ is the inverse 
of the matrix $\{g_{ab}\}$ (for non-symmetric $g$ it is 
the transposed inverse). Even though not relevant here, 
we remark that all relations in this paragraph 
are deliberately written in a way that remains valid 
for non-symmetric $g$. 
To end this paragraph we also remark that if the 
metric with respect to which index lowering and raising 
is clear from the context, one often writes (so-calles
``musical isomorphisms'') 
\begin{subequations}
\label{eq:DefMusicalIsomorphisms}
\begin{alignat}{2}
\label{eq:DefMusicalIsomorphisms-a}
&v^\flat&&\,:=\,\gdown(v)=g(v,\cdot)\,,\\
\label{eq:DefMusicalIsomorphisms-b}
&\alpha^\sharp&&\,:=\,\gup(\alpha)=g^{-1}(\cdot,\alpha)\,.
\end{alignat}
\end{subequations}
We make use of the notation 
\eqref{eq:DefMusicalIsomorphisms-a}
throughout the main text.

The inner product $g^{-1}$ on $V^*$ extends to an inner product 
on each $T^pV^*$ by 
\begin{equation}
\label{eq:DefInnProdTensors}
\bigl\langle
\alpha_1\otimes\cdots\otimes \alpha_p,\beta_1\otimes\cdots\otimes \beta_p
\bigr\rangle:=
\prod_{a=1}^p g^{-1}(\alpha_a,\beta_a)\
\end{equation}
and bilinear extension:
\begin{equation}
\label{eq:DefInnProdTensorsLinExt}
\bigl\langle
\alpha_{a_1\cdots a_p}\ \theta^{a_1}\otimes\cdots\otimes\theta^{a_p}\,,\,
\beta_{b_1\cdots b_p}\ \theta^{b_1}\otimes\cdots\otimes\theta^{b_p}\bigr\rangle=
\alpha_{a_1\cdots a_p}\beta^{a_1\cdots a_p}\,.
\end{equation}
In particular, it extends to each subspace $\Wedge^p V^*\subset T^pV^*$. 
We have 
\begin{equation}
\label{eq:InnProdPureWedges}
\bigl\langle
\alpha_1\wedge\cdots\wedge \alpha_p\,,\,\beta_1\wedge\cdots\wedge\beta_p
\bigr\rangle
:=p!\sum_{\sigma\in S_p}\mathrm {sign}(\sigma)\prod_{a=1}^p g^{-1}(\alpha_a,\beta_{\sigma(a)})
\end{equation}
and hence 
\begin{equation}
\label{eq:InnProdGenForms}
\Bigl\langle
\tfrac{1}{p!}\alpha_{a_1\cdots a_p}\theta^{a_1}\wedge\cdots\wedge\theta^{a_p}\,,\,
\tfrac{1}{p!}\beta_{b_1\cdots b_p}\theta^{b_1}\wedge\cdots\wedge\theta^{b_p}
\Bigr\rangle
=\alpha_{a_1\cdots a_p}\beta^{a_1\cdots a_p}\,.
\end{equation}

In the totally antisymmetric case it is more convenient to 
renormalise this product in a $p$-dependent fashion. One 
sets 
\begin{equation}
\label{eq:RenormInnerProduct}
\bigl\langle\cdot\,,\,\cdot\big\rangle_{\mathrm{norm}}\big\vert_{\Wedge^pV^*}
:=\tfrac{1}{p!}\,\bigl\langle\cdot\,,\,\cdot\big\rangle\big\vert_{\Wedge^pV^*}
\end{equation}
so that 
\begin{equation}
\label{eq:RenormInnProdGenForms}
\Bigl\langle
\tfrac{1}{p!}\alpha_{a_1\cdots a_p}\theta^{a_1}\wedge\cdots\wedge\theta^{a_p}\,,\,
\tfrac{1}{p!}\beta_{b_1\cdots b_p}\theta^{b_1}\wedge\cdots\wedge\theta^{b_p}
\Bigr\rangle_{\mathrm{norm}}
=\tfrac{1}{p!}\alpha_{a_1\cdots a_p}\beta^{a_1\cdots a_p}\,.
\end{equation}

\subsection{Hodge duality}
\label{sec:Appendix-HodgeDuality}
Given a choice $o$ of an orientation of $V^*$ (e.g. induced by an 
orientation of $V$), there is a unique top-form 
$\varepsilon\in\Wedge^nV^*$ (i.e. a \emph{volume form} for 
$V$), associated with the triple $(V^*,g^{-1},o)$, given by 
\begin{equation}
\label{eq:TopForm}
\varepsilon:=\theta^1\wedge\cdots\wedge\theta^n\,,
\end{equation}
where $\{\theta^1,\cdots,\theta^n\}$ is any 
$g^{-1}$-orthonormal Basis of $V^*$ in the 
orientation class $o$. The \emph{Hodge duality} 
map at level $0\leq p\leq n$ is a linear isomorphism 
\begin{subequations}
\label{eq:DefHodgeDuality}
\begin{equation}
\label{eq:DefHodgeDuality-a}
\star_p:\Wedge^pV^*\rightarrow\Wedge^{n-p}V^*\,,
\end{equation}
defined implicitly by 
\begin{equation}
\label{eq:DefHodgeDuality-b}
\alpha\wedge\star_p \beta=\varepsilon\,\langle\alpha\,,\,
\beta\rangle_{\mathrm{norm}}\,.
\end{equation}
\end{subequations}
This means that the image of $\beta\in\Wedge^pV^*$ under 
$\star_p$ in $\Wedge^{n-p}V^*$ is defined by the 
requirement that \eqref{eq:DefHodgeDuality-b} holds 
true for all $\alpha\in\Wedge^pV^*$. Linearity is immediate 
and uniqueness of $\star_p$ follows from the fact that 
if $\lambda\in\Wedge^{n-p}V^*$ and $\alpha\wedge\lambda=0$ for all 
$\alpha\in\Wedge^pV^*$, then $\lambda=0$. To show existence it 
is sufficient to define $\star_p$ on basis vectors. 
Since \eqref{eq:DefHodgeDuality-b} is also linear in 
$\alpha$ it is sufficient to verify \eqref{eq:DefHodgeDuality-b} 
if $\alpha$ runs through all basis vectors. 

From now on we shall follow standard practice and drop the 
subscript $p$ on $\star$, supposing that this will not cause 
confusion.  

Let $\{e_1,\cdots e_n\}$ be a basis of $V$ and 
$\{\theta^1,\cdots ,\theta^n\}$ its dual basis of $V^*$;
i.e. $\theta^a(e_b)=\delta^a_b$. Let further 
$\{\theta_{1},\cdots ,\theta_n\}$ be the basis of $V^*$
given by the image of $\{e_1,\cdots e_n\}$ under 
$\gdown$, i.e. $\theta_a=g_{ab}\theta^b$. Then, on the basis 
$\{\theta_{a_1}\wedge\cdots\wedge\theta_{a_p}\mid 1\leq a_1<a_2<\cdots <a_p\leq n\}$
of $\Wedge^pV^*$ the map $\star$ has the simple form   
\begin{equation}
\label{eq:StarMapOnBasis1}
\star(\theta_{b_1}\wedge\cdots\wedge\theta_{b_p})=\tfrac{1}{(n-p)!}
\varepsilon_{b_1\cdots b_p\,a_{p+1}\cdots a_n}\,\theta^{a_{p+1}}\wedge\cdots\wedge\theta^{a_n}\,.
\end{equation}
This is proven by merely checking \eqref{eq:DefHodgeDuality-b}
for $\alpha=\theta^{a_1}\wedge\cdots\wedge\theta^{a_p}$ and 
$\beta=\theta_{b_1}\wedge\cdots\wedge\theta_{b_p}$. Instead of 
\eqref{eq:StarMapOnBasis1} we can write
\begin{alignat}{1}
\label{eq:StarMapOnBasis2}
\star(\theta^{a_1}\wedge\cdots\wedge\theta^{a_p})
=&\tfrac{1}{(n-p)!}\
g^{a_1b_1}\cdots g^{a_pb_p}\
\varepsilon_{b_1\cdots b_pb_{p+1}\cdots b_n}\,
\theta^{b_{p+1}}\wedge\cdots\wedge\theta^{b_n}\nonumber\\
=&\tfrac{1}{(n-p)!}\
\varepsilon^{a_1\cdots a_p}_{\phantom{a_1\cdots a_p}a_{p+1}\cdots a_n}
\ \theta^{a_{p+1}}\wedge\cdots\wedge\theta^{a_n}\,,
\end{alignat}
which makes explicit the dependence on $\varepsilon$ 
and $g$.

If $\alpha=\frac{1}{p!}\alpha_{a_1\cdots a_p}\theta^{a_1}\wedge\cdots\wedge
\theta^{a_p}$, then $\star\alpha=\tfrac{1}{(n-p)!}
(\star\alpha)_{b_1\cdots b_{n-p}}
\theta^{b_1}\wedge\cdots\wedge\theta^{b_{n-p}}$, where  
\begin{equation}
\label{eq:StarInCoordinates}
(\star\alpha)_{b_1\cdots b_{n-p}}=
\tfrac{1}{p!}\,\alpha_{a_1\cdots a_p}
\varepsilon^{a_1\cdots a_p}_{\phantom{a_1\cdots a_p}b_1\cdots b_{n-p}}\,.
\end{equation}
This gives the familiar expression of Hodge duality in component 
language. Note that on component level the first (rather than 
last) $p$ indices are contracted.

Applying $\star$ twice (i.e. actually $\star_{(n-p)}\circ\star_p$)
leads to the following self-map of $\Wedge^pV^*$:
\begin{equation}
\label{eq:StarSquared}
\begin{split}
\star\bigl(&\star(\theta^{a_1}\wedge\cdots\wedge\theta^{a_p})\bigr)\\
&=\tfrac{1}{p!(n-p)!}
\varepsilon^{a_1\cdots a_p}_{\phantom{a_1\cdots a_p}a_{p+1}\cdots a_n}
\varepsilon^{a_{p+1}\cdots a_n}_{\phantom{a_{p+1}\cdots a_n}b_1\cdots b_p}\,
\theta^{b_1}\wedge\cdots\wedge\theta^{b_p}\\
&=\tfrac{(-1)^{p(n-p)}}{p!(n-p)!}\,
\varepsilon^{a_1\cdots a_pa_{p+1}\cdots a_n}
\varepsilon_{b_1\cdots b_pa_{p+1}\cdots a_n}\,
\theta^{b_1}\wedge\cdots\wedge\theta^{b_p}\\
&=(-1)^{p(n-p)}\,\langle\varepsilon,\varepsilon\rangle_{\mathrm{norm}}\
\theta^{a_1}\wedge\cdots\wedge\theta^{a_p}\,.
\end{split}
\end{equation}
Note that 
\begin{equation}
\label{eq:EpsilonSquared}
\langle\varepsilon,\varepsilon\rangle_{\mathrm{norm}}
=\tfrac{1}{n!} g^{a_1b_1}\cdots g^{a_nb_n}
\varepsilon_{a_1\cdots a_n}\varepsilon_{b_1\cdots b_n}
=(\varepsilon_{12\cdots n})^2/\det\{g(e_a,e_b)\}\,.
\end{equation}
This formula holds for any volume form $\varepsilon$ 
in the definition \eqref{eq:DefHodgeDuality-b}, 
independent of whether or not it is related to $g$. 

Since the right-hand side of \eqref{eq:DefHodgeDuality-b}
is symmetric under the exchange $\alpha\leftrightarrow\beta$,
so must be the left-hand side. Using \eqref{eq:StarSquared} 
we get 
\begin{equation}
\label{eq:StarStarInnProd1}
\begin{split}
\langle\alpha,\beta\rangle_{\mathrm{norm}}\,\varepsilon
&=\alpha\wedge\star\beta
=\beta\wedge\star\alpha
=(-1)^{p(n-p)}\star\alpha\wedge\beta\\
&=\langle\varepsilon,\varepsilon\rangle_{\mathrm{norm}}^{-1}\,
 \star\alpha\wedge\star\star\beta
=\langle\varepsilon,\varepsilon\rangle_{\mathrm{norm}}^{-1}\,
 \langle\star\alpha\,,\,\star\beta\rangle_{\mathrm{norm}}\,\varepsilon\,,
\end{split}
\end{equation}
hence
\begin{equation}
\label{eq:StarStarInnProd2}
\langle\star\alpha\,,\,\star\beta\rangle_{\mathrm{norm}}
=\langle\varepsilon,\varepsilon\rangle_{\mathrm{norm}}
 \langle\alpha,\beta\rangle_{\mathrm{norm}}\,.
\end{equation}
From this and \eqref{eq:StarSquared} it follows for 
$\alpha\in\Wedge^pV^*$ and $\beta\in\Wedge^{n-p}V^*$, 
that
\begin{equation}
\label{eq:StarAdjoint}
\langle\alpha,\star\beta\rangle_{\mathrm{norm}}
=\langle\varepsilon,\varepsilon\rangle_{\mathrm{norm}}^{-1}
\langle\star\alpha\,,\,\star\star\beta\rangle_{\mathrm{norm}}     
=\,(-1)^{p(n-p)}\,\langle\star\alpha,\beta\rangle_{\mathrm{norm}}\,.
\end{equation}
This shows that the adjoint map of $\star$ relative 
to $\langle\cdot\,,\,\cdot\rangle_{\mathrm{norm}}$ is 
$(-1)^{p(n-p)} \,\star$.  

Formulae \eqref{eq:StarSquared}, \eqref{eq:StarStarInnProd1}\eqref{eq:StarStarInnProd2}, 
and \eqref{eq:StarAdjoint} are valid for general $\varepsilon$ in the definition 
\eqref{eq:DefHodgeDuality-b}. If we choose $\varepsilon$ 
as the unique volume form that assigns unit volume to an 
oriented orthonormal frame, as does \eqref{eq:TopForm}, 
then we have 
\begin{equation}
\label{eq:VolumeFormSquared}
\langle\varepsilon,\varepsilon\rangle_{\mathrm{norm}}=(-1)^{n_-}
\end{equation}
where $n_-$ is the maximal dimension of subspaces 
in $V$ restricted to which $g$ is negative definite;
i.e. $g$ is of signature $(n_+,n_-)$. Equation 
\eqref{eq:StarStarInnProd2} then shows that $\star$ 
is an isometry for even $n_-$ and an anti-isometry for 
odd $n_-$. If $g$ is a Lorentzian inner product 
in $n$ dimensions then either $n_+=(n-1)$ and $n_-=1$,
which is called the ``mostly-plus convention'', or 
$n_+=1$ and $n_-=(n-1)$, which accordingly is called
the ``mostly-minus convention''. If $n$ is even $n_-$
is odd in both conventions, and no differences in signs 
occur. In contrast, for odd space-time dimensions all 
signs involving $(-1)^n_-$ differ in these two conventions.
In the mostly-minus convention that we followed in the main 
text, we have then 
$\star_{n-p}\circ\star_p=(-1)^{(n+1)(p+1)}\id\big\vert_{\Wedge^pV^*}$, 
where we used the following identity for integers mod(2):
$p(n-p)+(n-1)=np+p+n+1=(n+1)(p+1)$. This means that in even 
space-time dimensions (like the $n=4$ we are mostly interested 
in) $\star$ squares to the identity on odd forms
and to minus the identity on even forms, whereas in odd space-time
dimensions it always squares to the identity. As explained above,
this remains true for even $n$ in the mostly-plus convention,
but turns to the opposite signs for odd $n$ if the conventions
are changed.

Finally we note the following useful formula: If $v\in V$
let $i_v: T^pV^*\rightarrow T^{p-1}V^*$ the map which inserts 
$v$ into the first tensor factor. It restricts to a 
map $i_v: \Wedge^pV^*\rightarrow\Wedge^{p-1}V^*$. 
Then, for any $\alpha\in\Wedge^pV^*$, we have   
\begin{equation}
\label{eq:InsertionAndHodge}
i_v\star\alpha=\star(\alpha\wedge v^\flat)\,.
\end{equation}
using the notation \eqref{eq:DefMusicalIsomorphisms-a}.
It suffices to prove this for basis elements $v=e_a$ of $V$ 
and $\alpha=\theta^{a_1}\wedge\cdots\wedge\theta^{a_p}$ of $\Wedge^pV^*$, 
which is almost immediate using \eqref{eq:StarMapOnBasis2}.

An immediate consequence of \eqref{eq:InsertionAndHodge}
is that for any $v\in V$ we have the identity 
$i_v\star v^\flat=\star(v^\flat\wedge v^\flat)=0$ and hence that 
$i_v(v^\flat\wedge\star v^\flat)=g(v,v)\,\star v^\flat$.
On the other hand, from \eqref{eq:DefHodgeDuality-b} the 
left hand side of the last equation also equals to 
$\langle v^\flat,v^\flat\rangle i_v\varepsilon$.
Since $g(v,v)=\langle v^\flat,v^\flat\rangle_{\mathrm{norm}}$ we have 
\begin{equation}
\label{eq:HodgeOfVectors}
\star v^\flat=i_v\varepsilon\,.
\end{equation}

\newpage

\end{document}